\RequirePackage[log]{snapshot}
\def\year{2021}\relax
\documentclass[letterpaper]{article} 

\usepackage{modified-aaai21}  
\usepackage{times}  
\usepackage{helvet} 
\usepackage{courier}  
\usepackage[hyphens]{url}  
\usepackage{graphicx} 
\usepackage{natbib}  
\usepackage{caption} 
\usepackage[colorlinks=true, citecolor=blue]{hyperref}
\usepackage{url}            
\usepackage{booktabs}       
\usepackage{amsfonts}       
\usepackage{nicefrac}       
\usepackage{microtype}      
\usepackage{amsthm}
\usepackage{amsmath}
\usepackage{mathtools}
\usepackage{comment}
\usepackage{textcomp}
\usepackage{gensymb} 
\usepackage{enumitem}
\usepackage{xcolor}

\usepackage{amssymb}
\usepackage{graphics}
\usepackage{graphicx}
\usepackage{color}

\newtheorem{theorem}{Theorem}
\newtheorem{corollary}[theorem]{Corollary}

\newtheorem{proposition}[theorem]{Proposition}
\newtheorem{lemma}[theorem]{Lemma}

\newtheorem{definition}[theorem]{Definition}
\theoremstyle{remark}
\newtheorem{example}[theorem]{Example}

\newcommand{\R}{\mathbb R}

\renewcommand{\P}{\mathbb{P}}

\newcommand{\F}{\mathcal{F}}

\newcommand{\X}{\mathcal{X}}
\newcommand\abs[1]{\left| #1  \right|}

\newcommand{\E}{\mathbb E}

\newcommand{\GE}{\mathrm{WE}}
\newcommand{\UC}{c_\textnormal{unc}^*}

\newcommand{\fed}{}

\newcommand{\omer}{\color{red}}

\newcount\Putacknowledgement 
\newcount\Includeappendix 
\newenvironment{proofof}[1]{\begin{proof}[\textnormal{\textbf{Proof of #1}}]}{\end{proof}}
\newenvironment{proofsketchof}[1]{\begin{proof}[\textnormal{\textbf{Proof sketch of #1}}]}{\end{proof}}

\title{Protecting the Protected Group: Circumventing Harmful Fairness}
\author{

    Omer Ben{-}Porat\textsuperscript{\rm 1}, Fedor Sandomirskiy\textsuperscript{\rm 2,\rm 3}, Moshe Tennenholtz\textsuperscript{\rm 2}\\
}
\affiliations{
    \textsuperscript{\rm 1}Tel-Aviv University\\
    \textsuperscript{\rm 2}Technion---Israel Institute of Technology\\
    \textsuperscript{\rm 3}Higher School of Economics, St.~Petersburg, Russia\\
    \texttt{{omerbenporat@mail.tau.ac.il}, {fedor.sandomirskiy@gmail.com}, {moshet@ie.technion.ac.il}}

}

\begin{document}

\maketitle

\begin{abstract}

The recent literature on fair Machine Learning manifests that the choice of fairness constraints must be driven by the utilities of the population. However, virtually all previous work makes the unrealistic assumption that the exact underlying utilities of the population (representing private tastes of individuals) are known to the regulator that imposes the fairness constraint. In this paper we initiate the discussion of the \emph{mismatch}, the unavoidable difference between the underlying utilities of the population and the utilities assumed by the regulator. We demonstrate that the mismatch can make the disadvantaged protected group worse off after imposing the fairness constraint and provide tools to design fairness constraints that help the disadvantaged group despite the mismatch.


\end{abstract}	

\section{Introduction}
At first glance, algorithms may seem free of human biases such as sexism or racism. However, in many situations, they are not: the automated recruiting tool used by Amazon was favoring men \cite{doleac2016does}; judges in the US use the COMPAS algorithm to estimate the probability that the defendant will re-offend while this algorithm was accused of being biased against black people~\cite{larson2016we}.  See \citet{o2016weapons} for many more examples. These challenges call for imposing fairness constraints on algorithm design and, in particular, on machine-learned classifiers, which are the subjects of this paper. As a running example, consider a bank that gives loans to potential borrowers and is regulated by a policy-maker. The bank learns a decision rule (namely, a classifier) from historical data to decide for whom to approve or decline the loan to maximize its revenue that is increasing with the number of repaid loans. As repeatedly observed in the past, the resulting classifier may be biased against a protected group (e.g., ethnic minority). Hence, the regulator may wish to impose a fairness constraint on the bank.

Bias is deemed unjust. Beyond that, it affects the \textit{welfare} of protected groups, as borrowers have preferences towards the different outcomes of the classification, captured by utility functions.\footnote{We use the term utility for the well-being of a single individual and welfare for the aggregated well-being of groups of individuals.}
Consequently, the goal of imposing fairness constraints is to improve the welfare of the \textit{disadvantaged} group (the originally-discriminated one). A natural approach by which the  regulator can achieve this goal is by assuming a utility function and requiring the bank's classifier to equalize the protected groups' welfare. 
The two most popular fairness constraints, \emph{Demographic Parity}  \cite{agarwal2018reductions,dwork2012fairness}  and \emph{Equal Opportunity} \cite{hardt2016equality} (henceforth DP and EO, respectively) are special cases of this approach for particular utility functions. For instance, DP is recovered by assuming that every agent gets a utility of one when receiving the loan, and zero otherwise.

The possibility that fairness may harm the well-being of those it is designed to protect, i.e., that it can harm the disadvantaged group's welfare, seems counter-intuitive. However, it is well known both theoretically and empirically that the most intuitive fairness constraint of Unawareness (which forbids using the sensitive attribute in classification) can be harmful \cite{corbett2018measure,dwork2018decoupled,ustun2019fairness,doleac2016does}. For example, \citet{doleac2016does} show that the ``ban the box'' policy, adopted by the United States and preventing employers from seeing applicants' criminal background, decreased the welfare of discriminated minorities (the chances of getting a job). Additionally, \citet{liu2018delayed} discovered that DP and EO are not free of the same flaw if we consider long-term consequences. They show that fairness constraints may force the bank to give loans to those members of the disadvantaged group who, otherwise, would not have received the loans due to the high probability of default. Therefore, such unqualified borrowers are likely to have problems with paying back the loan. This increased default ratio would harm the disadvantaged group's average credit score, thereby harming its welfare in the long run.

\subsection{Our contribution}
In this paper, we uncover another mechanism underlying harmful fairness even in static settings:

\begin{quote}
\emph{Imposing a fairness constraint can make the disadvantaged group worse off if the fairness constraint and the utilities of the population mismatch.}
\end{quote}
Following the recent trends of fair ML literature, we assume that agents may have different preferences over classification outcomes, which are captured by utility functions. For example, borrowers may differ in their value for getting the loan depending on their access to alternative sources of money and on the purpose of borrowing. With utilities in hand, we can use social welfare to evaluate a group's well-being for any given classifier.
To talk about the mismatch between utilities and fairness constraints, the latter has to be defined in utilitarian terms. As we described above, fairness constraints like DP and EO are naturally cast as equalizing \textit{some} welfare functions of the protected groups. 
However, the difficulty in applying any welfare-based fairness constraint is that the utilities must be known to the designer, while these utilities represent individuals' private tastes. Hence, even if domain experts determine the utilities used in a fairness constraint (henceforth, \emph{assumed utilities}), they can only approximate the actual utilities of the population (\emph{underlying utilities} in what follows). This discrepancy leads to the following conclusion.
\begin{quote}
\emph{In practice, the mismatch between the  underlying utilities of the population and the utilities assumed by the regulator is unavoidable.} 
\end{quote}
Together with the observation that the mismatch can make fairness harmful, this becomes a serious caution for regulators that design fairness constraints for a certain industry, e.g., banking. We complement this caution with a positive message. Naturally, a small discrepancy between underlying and assumed utilities must be innocuous. However, we characterize a much more applicable and promising connection; we show that
\begin{quote}
\emph{Fairness constraints help the disadvantaged group whenever the utilities assumed by the regulator and the underlying utilities of the population agree on which group is disadvantaged}.
\end{quote}
Finally, we suggest additional ways to deal with the mismatch if the underlying utilities can be approximated from data.

\subsection{Related Work on  Economic Ideas in Fair Classification}
Welfare-Equalizing, our approach to fairness, has a long history in normative economics~\cite{pazner1978egalitarian,roemer1986equality} (where it is known under the name of egalitarianism) and political philosophy~\cite{rawls2009theory}; it was used for fair resource allocation without money transfers~\cite{li2013egalitarian}, in the field of cooperative games~\cite{dutta1989concept} and bargaining problems~\cite{kalai1975other}. In contrast to recent papers on the utilitarian approach to fair classification \cite{heidari2018fairness,heidarifat19}, which suggest maximizing the minimal welfare among the protected groups, we strengthen this desideratum by making it a normative requirement: the welfare must be equal among the subgroups defined by a sensitive attribute. This normative condition allows one to separate the \emph{fairness constraint} (which may be imposed by a regulator) from the \emph{selfish objective of the decision-maker} (a revenue-maximizing bank in our running example) and thus allows one to analyze how decisions change after imposing the fairness constraint. Another advantage of the Welfare-Equalizing concept is the simple threshold structure of the optimal fair classifier (similar to the one for Demographic Parity or Equal Opportunity \cite{corbett2017algorithmic}), which makes it efficiently computable.
%

This work joins recent attempts \cite{rambachan2020economic,rambachan2020economicA,HuFat2020,ElzaynFat2020,hossaindesigning} to bring better economic understanding to fairness in ML; we address some of them here and refer the reader to \citet{finocchiarofairness} for a comprehensive survey. \citet{HuFat2020} propose an optimization framework for fair classification and welfare analysis. In their modeling, a learner executes a soft-margin SVM with the additional constraint of group fairness: limiting the two groups' welfare discrepancy to a predefined quantity. They provide a sensitivity analysis, showing that applying stricter fairness constraints (decreasing the allowed discrepancy) can worsen welfare outcomes
for both groups. Their findings are in line with ours, but our analysis is fundamentally different; in particular, they do not address the utility mismatch issue. Imposing fairness constraints on profit-maximizing entities, as we do in this paper, is an understudied point of view, as noted recently by \citet{ElzaynFat2020}.

There are also several recent attempts to harness economic principles to fair classification. 
For example, \citet{golz2019paradoxes} treat fair classification as an allocation of goods, where there is a fixed amount of resources to distribute. They examine the compatibility (or lack thereof) of Equalized Odds with axioms of fairness from the economic literature on fair division; see~\cite{brandt2016handbook} for a survey. Envy-freeness, the dominant fairness concept in economics, plays a crucial role in several recent papers at the intersection of economics and AI \cite{caragiannis2019unreasonable, cohler2011optimal,benade2018make,guruswami2005profit, gal2016fairest, plaut2018almost},  including several works on fair classification \cite{zafar2017parity, balcan2019envy,ustun2019fairness,hossaindesigning}.
However, these papers on fair classification focus on sample complexity and generalization \cite{balcan2019envy}, or asserting that users favor treatment disparity \cite{ustun2019fairness,zafar2017parity} in health applications.

The work most related to ours is the paper by \citet{hossaindesigning}. Concurrently to and independently of our work, \citet{hossaindesigning} argue for group equability in fair classification, which coincides with our Welfare-Equalizing fairness constraint. They, too, show that their concept subsumes previously suggested fairness notions. However, there is a significant difference between the two works. First, \citet{hossaindesigning} are interested in learning the best classifier from data, and hence address issues of generalization from samples and differentiability. In contrast, we devote our paper to societal considerations of fair classification, and thus consider fairness as a post-processing step similarly to \cite{corbett2017algorithmic,hardt2016equality}. Additionally, in contrast to \citet{hossaindesigning}, our analysis focuses on the mismatch between the utilities assumed by the regulator and the actual, underlying utilities.

\subsection{Paper structure}

In Section \ref{sect_model}, we present our formal model, define the Welfare-Equalizing fairness framework, and prove structural results for optimal fair classifiers in Subsection \ref{subsect_optimal_GE}. Section \ref{sect_implications} 
deals with the implications of a mismatch between the population's underlying utilities and the utilities assumed by the regulator. Finally, Section~\ref{sect_computation} describes how to compute the bank-optimal fair classifier if the assumed utilities approximate the underlying utilities well enough.

\section{Model}\label{sect_model}
We consider a general classification problem, where agents have an ex-ante non-observable ``quality'' correlated with observable attributes. 
We keep using the metaphor of a bank that predicts the reliability of the population and makes  lending decisions; however, the same setting captures student admissions, recruiting, assessing the recidivism risk for a criminal, etc.

There are three parties in the model: a heterogeneous population of potential borrowers; a bank that makes lending decisions based on the observable attributes of  borrowers and cares only about its revenue; and a regulator that cares about fairness  and can restrict the set of lending policies available to the bank by imposing a fairness constraint.  We now present these parties formally.

\subsubsection{Borrowers} We assume that each potential borrower (henceforth borrower) is associated with a pair of observable attributes $(X,A)\in \X\times \{0,1\}$.  Here $A$ is a binary\footnote{The assumption of the dichotomy of $A$ is made for simplicity. Extending our results to the non-binary case (ethnicity) is straightforward.} sensitive attribute (e.g., gender) and $X\in \X$ encodes all other characteristics of a borrower, e.g., employment history, salary, education, assets and so on. We do not impose any assumptions on $\X$.  By $\{A=0\}$ and $\{A=1\}$ we denote the groups of all borrowers with the sensitive attribute equal to $0$ or $1$, respectively; we call $\{A=a\}$ a \emph{protected group}. Furthermore, in addition to the observable attributes $X$ and $A$, every borrower is also associated with an unobservable variable  $Y\in \{0,1\}$, which describes whether that borrower will pay back the loan or not. For brevity, we call borrowers with $Y=1$ and $Y=0$, \emph{good} and \emph{bad}, respectively. The statistical characteristics of the population are described by a probability space $(\Omega,\F,\P)$; so $X=X(\omega)$,   $A=A(\omega)$ and $Y=Y(\omega)$ are random variables on $\omega \in \Omega$. By small letters $(x,a,y)$ we denote realizations of $X$, $A$, and $Y$, i.e., generic elements of $\X\times \{0,1\}\times\{0,1\}$.

Each borrower $(x,a,y)$ obtains a \textit{utility} when receiving the loan. This utility can depend on $x$ and $a$ in many ways, but what is more critical is that it must depend on the non-observable quality of the borrower, namely $v=v(x,a,y)$. To simplify the presentation, we assume that the utility from a rejected application is zero and that the average utility of a borrower with given $x$ and $a$ is non-negative,\footnote{Zero utility for not getting a loan is a normalization-condition: before borrowing money, everybody is at zero utility level. Non-negativity of $v$ can be regarded as a rationality assumption on borrowers: no rational agent would apply for a loan if she/he expects that getting the loan brings negative utility while not getting gives $0$.}  i.e., $\E[v(X,A,Y)\mid X=x,A=a ]\geq 0$. 
However, both assumptions can be relaxed. We refer to $v$ as the \emph{underlying utility} of the population.

\subsubsection{The bank} We assume that the bank knows the joint distribution of $(X,A,Y)$ from historical data. In particular, it knows the exact conditional probability of being a good borrower given the observable attributes; we denote it by $p(x,a)=\P(Y=1\mid X=x, A=a)$.\footnote{Indeed, this is aligned with previous works that consider fairness as a post-processing step \cite{corbett2017algorithmic,hardt2016equality}.}

The bank makes lending decisions based on $X$ and $A$ but without observing $Y$. It uses a classifier $c: \X\times \{0,1\} \to [0,1]$ where $c(x,a)$ is the probability of giving a loan to a population of borrowers with  $X=x$ and $A=a$. Each loan given to a good borrower brings a revenue\footnote{In contrast to the rest of the literature, we allow the  bank's revenue to  depend on the non-sensitive attribute $X$. This becomes important if $X$ also encodes the type of loan a client is applying for, e.g., different borrowers may need a different amount of money and thus bring a different revenue/loss.} of $\alpha_+(X)>0$ to the bank while each bad borrower leads to a loss of $\alpha_-(X)>0$; we assume that $\alpha_\pm(x)$ are bounded functions of $x\in \X$. The bank's revenue depends on the choice of a classifier $c$, and is defined by
\begin{equation}\label{eq_objective}
R(c)=\E\left[c(X,A)\left(  \alpha_+(X) Y - \alpha_-(X)(1-Y)  \right)\right].
\end{equation}
To ease notation, we define $t(x)$ and $r(x,a)$ such that for every $a\in \{0,1\}, x\in \X$
\begin{align}
& t(x)\coloneqq\frac{\alpha_-(x)}{ \alpha_+(x)+\alpha_-(x)}, \label{eq_t(x)}\\ 
& r(x,a)\coloneqq\left(\alpha_+(x)+\alpha_-(x)\right) \left( p(x,a)-t(x)\right); \label{eq_r(x)}
\end{align}
hence, we can  rewrite Equation \eqref{eq_objective} by taking a conditional  expectation with respect to $A,X$ as
{\thinmuskip=0mu
\medmuskip=0mu plus 0mu minus 2mu
\thickmuskip=2mu plus 2mu
\begin{align}\label{eq_objective_with_t}
R(c)&=\E\left[c(X,A)\big(\alpha_+(X) p(X,A)-\alpha_-(X)(1-p(X,A)\big) \right]\nonumber\\
&=\E\left[r(X,A) c(X,A) \right].
\end{align}}%
The goal of the bank is to maximize $R(c)$ over the set of \textit{feasible} classifiers, i.e., classifiers that satisfy the regulator's constraints. 
\subsubsection{The regulator}\label{subsect_GE}
The regulator evaluates the well-being of a group using its \emph{welfare}: the expected utility of its members.
For an underlying utility function $v$ and a classifier $c$, the welfare of the subgroup $\{A=a\}$ is given by
\begin{equation}\label{eq_welfare_def}
 W_{v,c}(a)=\E\left[v(X,A,Y)c(X,A)\mid A=a\right].~
\end{equation}
The regulator aims to equalize welfare among the protected groups. However, as we discuss in the introduction, the underlying utility $v$ is unknown to the regulator; thus, the regulator is forced to use a certain substitute $u$ instead. We refer to $u$ as the \emph{assumed utility}; ideally, the assumed utility must be an approximation to the underlying one.

The objective of the regulator is captured by the following $u$-\textit{Welfare-Equalizing} constraint ($u$-WE for abbreviation).
\begin{definition}[$u$-WE classifier]\label{def:equalizing}
Given a utility-function $u$, a classifier $c$ is \textit{$u$-Welfare-Equalizing} if
\begin{equation}\label{eq: u-WE}
W_{u,c}(0)=W_{u,c}(1),
\end{equation}
i.e., if $c$ equalizes the welfare among the two protected groups. The set of all such classifiers is denoted by  $\GE(u)$.
\end{definition}
Note that Welfare-Equalizing constraint is defined with respect to the assumed utility.

\subsection{Special cases of Welfare-Equalizing fairness}\label{subsec_WE_contains_DP_EO}
The framework of {Welfare-Equalizing} fairness
allows one to analyze existing fairness constraints in a unified manner. For instance,
\begin{itemize}[leftmargin=*]
    \item The fairness constraint of \emph{Demographic Parity} (DP) (e.g., \cite{agarwal2018reductions,dwork2012fairness}) requires that the fraction of those who receive loans in the two groups must be the same. Formally, a classifier $c$ satisfies DP if $\E[c(X,A)\mid A=0]=\E[c(X,A)\mid A=1]$. It is a special case of WE fairness with $u(x,a,y)\equiv 1$ for any triplet $(x,a,y)$.
    \item Motivated by drawbacks of DP, \citet{hardt2016equality} concluded that good and bad borrowers within protected groups must be treated separately and introduced the concept of \emph{Equal Opportunity} (EO). Under this fairness constraint, the fraction of good borrowers who get loans must be the same in the two subgroups. Formally, a classifier $c$ satisfies EO if 
\begin{equation*}
\thinmuskip=1mu
\medmuskip=1mu plus 1mu minus 1mu
\thickmuskip=4mu plus 4mu
\E\left[c(X,A)\mid Y=1,A=0\right]= \E\left[c(X,A)\mid Y=1,A=1\right].
\end{equation*}
We recover EO by setting $u(x,a,y)= y\cdot \beta_a$. The coefficients $\beta_a=\nicefrac{1}{\E[Y\mid A=a]}$ normalize the maximal possible welfare in each group to~$1$. Such a rescaling is known under the name of ``relative welfare'' and is commonly used in economics to make welfare or utilities among groups comparable~\cite{kalai1975other}.
    \item Borrowers can differ in the amount of money $m$ they need. We can assume that information about $m$ is encoded in $X$, so $m=m(X)$. Then, a straightforward generalization of EO is the following concept of  \emph{Heterogeneous}-EO given by
    $u(x,a,y)=  y\cdot m(x)\cdot \beta_a$ with $\beta_a=\nicefrac{1}{\E[m(x)\mid A=a]}$.
    We can capture any other heterogeneity similarly (e.g., different interest rates, time-period, and payment schedules).
\end{itemize}

\subsection{Structural properties of the bank-optimal classifiers}\label{subsect_optimal_GE} 
In this subsection, we analyze classifiers that are optimal for the bank. We first state the result of \citet{corbett2017algorithmic}, who characterize the structure of the bank-optimal \textit{unconstrained} classifier. Then, we build upon their results for the \textit{constrained} case. Namely, we assume that the regulator imposes the $u$-WE fairness constraint on the bank and explore the structure of the bank-optimal classifiers. We show that the optimal classifiers have a generalized threshold structure, a fact that is extensively used in Sections~\ref{sect_implications} and~\ref{sect_computation}.

\subsubsection{Unconstrained classifier $\UC$}
If the regulator imposes no constraint on the bank, i.e., the bank is free to choose any classifier, then the revenue-maximizing classifier  has a simple form \cite{corbett2017algorithmic}.  Only borrowers with $r(x,a)>0$ are profitable for the bank, which is equivalent to the probability of paying back $p(x,a)$ being greater than $t(x)$ (recall that $t$ and $r$ are defined by Equations~\eqref{eq_t(x)} and~\eqref{eq_r(x)}). Consequently, the optimal lending policy is given by the following threshold classifier $\UC$: all borrowers with $p(x,a)>  t(x)$ get loans ($\UC(x,a)=1$) and all borrowers with $p(x,a)\leq t(x)$ are rejected ($\UC(x,a)=0$).\footnote{For definiteness, we assume that if the bank finds the two decisions equally profitable (the knife-edge case $p(x,a)=t(x)$), it chooses  the one with fewer loans given (e.g., this policy minimizes paperwork).}

\subsubsection{Constrained classifier $c^*_{\GE(u)}$}
Consider the general, constrained case, where the regulator imposes $u$-WE fairness on the bank. For a fixed and given assumed utility $u$, we denote by $c^*_{\GE(u)}$ the classifier that maximizes the bank's revenue $R(c)$ (see Equation \eqref{eq_objective}) among all $u$-WE classifiers $c\in \GE(u)$. The set of $u$-WE classifiers is non-empty since $0\in \GE(u)$ and, therefore, the bank's optimization problem is well-defined. 
To ease notation, we denote by $\overline u(x,a)$ the assumed utility of a borrower associated with $(x,a)$ averaged over the possible values of $Y$,
\begin{align*}
\overline u (x,a)&=\E[u(X,A,Y)\mid X=x,A=a ]\\
&=u(x,a,1)p(x,a)+u(x,a,0)(1-p(x,a)).    
\end{align*}
Further,  we denote by $R_a^*(w)$ the maximal revenue that the bank could extract from the group $\{A=a\}$ at the (assumed) welfare level $W_{u,c}(a)=w$. Formally,
\[
R_a^*(w)=\max_{\{c:\, \X\to[0,1]\, \mid\, W_{u,c}(a)=w\}} \E\left[r(X,A)\cdot c(X)\mid A=a\right],
\]
where $r$ is given by Equation~\eqref{eq_r(x)}.
The following Proposition \ref{prop_optimal_GE} shows that the bank-optimal constrained classifier always exists and reveals its structure.
\begin{proposition}\label{prop_optimal_GE}
The bank-optimal $u$-WE classifier $c^*_{\GE(u)}$ exists. Furthermore, each optimal classifier has the following form: 
\begin{equation}\label{eq_optGE_classifier}
c^*_{\GE(u)}(x,a)=
\begin{cases}
1     &  r(x,a)> \lambda_a \overline{u}(x,a)  \\
0     &  r(x,a)< \lambda_a \overline{u}(x,a) \\
\tau_a(x) & r(x,a) = \lambda_a \overline{u}(x,a)
\end{cases}.
\end{equation}
The group-dependent thresholds $\lambda_{a}$, $a\in\{0,1\}$  belong to the super-gradient\footnote{For a concave function $f=f(t), \ t\in [t_0,t_1]$, the super-gradient $\partial_{t} f $  is the set of all $q\in\R$ such that $f(t')\leq f(t)+q(t'-t)$ for all $t'$. If $f$ is continuous, then for any $t$ the super-gradient is non-empty, see \citet{rockafellar2015convex}.} of the subgroup-optimal revenue $R_a^*(w)$  (a concave function of $w$) computed at the welfare level $w^*$ maximizing the total bank's revenue $\P(A=0)R_0^*(w)+\P(A=1)R_1^*(w)$. The functions $\tau_{a}:\, \X\to [0,1]$ are arbitrary\footnote{In particular, there always exists $c^*_{\GE(u)}$ with constant $\tau_a$, i.e., independent of $x$.} up to the constraint that $c^*_{\GE(u)}$ provides the desired welfare level $w^*$ for both groups,   $w^*=W_{u,c^*_{\GE(u)}}(0)=W_{u,c^*_{\GE(u)}}(1)$.
\end{proposition}
\begin{proofsketchof}{Proposition \ref{prop_optimal_GE}}
The revenue maximization over $c\in \GE(u)$ can be represented as a two-stage procedure. In the first stage, we find the revenue-maximizing classifier in each of the subgroups $\{A=a\}$ given the welfare level $w$; in the second stage, we optimize over $w$. The welfare constraint in  the first stage can be internalized using the Lagrangian approach; the corresponding Lagrange multipliers $\lambda_a$ are equal to the ``shadow prices'', i.e., the derivatives of the value functions $R_a^*(w)$   with respect to $w$. This internalization reduces finding the subgroup optimal classifier to the unconstrained problem; thus, the optimal classifier has a threshold structure similarly to $\UC$. This structure is inherited by~$c^*_{\GE(u)}$. Since the resulting linear program is infinite-dimensional and $R_a^*(w)$ may be non-differentiable, the formal proof requires some functional-analytic arguments presented in the appendix.
\end{proofsketchof}
For the special cases of Demographic Parity and Equal Opportunity, the explicit form of the optimal classifiers was obtained  by \citet{corbett2017algorithmic}. Their result becomes an immediate corollary of Proposition~\ref{prop_optimal_GE}; see~{\ifnum\Includeappendix=1{Subsection~\ref{subsect_ex_threshold_DP_EO} in the appendix}\else{the appendix}\fi}.

\begin{figure*}[t!]
    \centering
    \includegraphics[scale=.7]{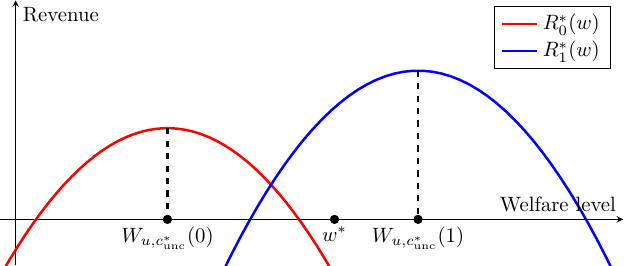}
    \caption{Intuition for the proof of Lemma~\ref{lm_no_mismatch}. The red curve illustrates the revenue $R_0^*(w)$ from the disadvantaged group $\{A=0\}$ at every possible welfare level, and the blue curve illustrates $R_1^*(w)$ from the advantaged group $\{A=1\}$. The subgroup revenues $R_a^*(w)$ are concave functions of $w$ that attain their maxima at a welfare level of $W_{u,\UC}(a)$. The total revenue is 
    $\P(A=0)R_0^*(w)+\P(A=1)R_1^*(w)$, which is maximized at $w^*=W_{u,c_{\GE(u)}^*}$. Noticeably, it always lies between the two maxima. In this illustration, $\P(A=0)=\frac{1}{3}$ and $\P(A=1)=\frac{2}{3}$.  }
    \label{fig:age}
\end{figure*}

\section{Mismatch of Fairness and Underlying Utilities}\label{sect_implications}
As we discussed in the introduction, the regulator aiming to equalize welfare among protected groups unavoidably assumes a certain approximation $u$ of the underlying utilities~$v$. For example, $u$ can be determined by the domain experts, while $v$ reflects the private tastes of the population and hence is not observable directly. We refer to the fact that $u$ is different from $v$ as a \emph{mismatch}. This section explores how imposing the Welfare-Equalizing fairness constraint with respect to $u$ affects the underlying welfare, which is measured by $v$.

We use the situation that exists before imposing the fairness constraint as the benchmark. We say that a group $\{A=a\}$ is $v$-\textit{disadvantaged} if under the bank-optimal unconstrained classifier, the welfare of $\{A=a\}$ is lower than the welfare of the other group $\{A=1-a\}.$
Formally, the group $\{A=a\}$ is $v$-disadvantaged if 
\[
W_{v,\UC}(a)< W_{v,\UC}(1-a),
\]
where $W$ is defined in Equation \eqref{eq_welfare_def} and $\UC$ is the bank-optimal unconstrained classifier from Subsection~\ref{subsect_optimal_GE}. When the underlying utility $v$ is clear from the context, we say that the group is disadvantaged and omit the dependence on $v$.

Ideally, imposing WE-fairness (or any other fairness constraint) should improve the welfare of the disadvantaged group. However, as we show next, this is not always the case. We say that a fairness constraint \emph{harms} the group $\{A=a\}$ if $W_{v,c^*}(a)< W_{v,\UC}(a)$,
where $c^*$ is the bank-optimal classifier after imposing that fairness constraint. Put differently, the fairness constraint harms the group $\{A=a\}$ if the welfare of the group (measured with respect to underlying utilities) decreases after imposing the fairness constraint.

\subsubsection{Harmful mismatch} 
We now demonstrate that the mismatch between assumed and underlying utilities can make the fairness constraint harmful to the disadvantaged group.
\begin{example}
Let the underlying utility-function be $ v(x,a,y)\equiv1$, i.e., all borrowers  equally benefit from receiving loans.
However, the regulator does not know the underlying utilities and decides to impose the fairness constraint of EO. Equivalently, the regulator assumes the utility $u(x,a,y)= y\cdot \beta_a$, for normalizing coefficients $\beta_a$ (as we explained in Subsection~\ref{subsec_WE_contains_DP_EO}). Notice that the underlying utility is the one associated with DP, and the regulator's assumed utility is the one associated with EO. Since $v\ne u$, there is a mismatch.

Suppose that $\X=\{0,1,2\}$, and all the combinations of $(x,a)$ have the same probability of $\frac{1}{6}$. Furthermore, assume that the fraction $p(x,a)$ of good borrowers is given by the table
\[
\begin{array} {c|ccc}
        & x=0 & x=1 & x=2 \\
        \hline
a=0 & 3/4 & 3/4 & 1/4   \\
a=1 & 1 & 0  & 0
\end{array}.
\]
In addition, let the revenue of the bank from a paid-back loan be {$\alpha_+(x)=1$,} and the loss from a borrower's default be {$\alpha_-(x)=2$} for every $x\in \X$.

We first want to determine which group is disadvantaged. The threshold for the bank-optimal unconstrained classifier $\UC$ equals {$t(x)= \frac{2}{3}$.} Hence, under $\UC$,  only borrowers with $x=0$ receive loans in the group $\{A=1\}$. In $\{A=0\}$, borrowers with $x\in\{0,1\}$ receive loans since in such cases {$t(x) = \frac{2}{3}<\frac{3}{4}=p(x,0)$;} however, loans are not given to borrowers with $x=2$ since {$t(2) = \frac{2}{3}>\frac{1}{4}=p(2,0)$.} Consequently, $\{A=1\}$ is disadvantaged: $W_{v,\UC}(0)=\frac{2}{3}$ compared to $W_{v,\UC}(1)=\frac{1}{3}$.

Next, let us examine how imposing the fairness constraint of EO changes the outcome of the bank-optimal classifier. The proportion of loans given by $\UC$ to good borrowers in $\{A=0\}$ is equal to the welfare of this group with respect to the assumed utility $u$, namely
\[
W_{u,\UC}(0)=\E\left[\UC(X,A)\mid Y=1,A=0\right]=\frac{6}{7}.
\]
In contrast, for $\{A=1\}$ we have
\[
W_{u,\UC}(1)=\E\left[\UC(X,A)\mid Y=1,A=1\right]=1.
\]
By imposing $u$-WE-fairness, the regulator requires the bank to equalize these two quantities. To do so, the bank-optimal constrained classifier can  either increase the amount of loans given to $\{A=0\}$ by approving some applications of $x=2$ or decrease the number of loans given to $\{A=1\}$. However, giving loans to  $x=2$ in $\{A=0\}$ is too costly: the cost $\frac{3}{4}\alpha_-  - \frac{1}{4}\alpha_+$ is not compensated by the benefit $1\cdot\alpha_+$ from giving the same amount of loans to good borrowers with  $x=0$ in  $\{A=1\}$. Therefore, the bank-optimal constrained classifier coincides with $\UC$ in $\{A=0\}$ and gives fewer loans  to $x=0$ in the group $\{A=1\}$: $c^*_{\GE( u)}(0,1)=\frac{6}{7}$. Consequently, the bank equalizes the proportion of loans given to good borrowers in both groups: $W_{ u,c^*_{\GE(u)}}(0)=W_{u,c^*_{\GE(u)}}(1)=\frac{6}{7}$. However, since the underlying utility is measured by $v$, not $u$, we get that $W_{v, c^*_{\GE(u)}}(1)=\frac{2}{7}<W_{v,\UC}(1)=\frac{1}{3}$ and the disadvantaged group is harmed by imposing the fairness constraint.
\end{example}

\subsection{Can WE-fairness help the disadvantaged group despite {the} mismatch?}\label{subsect_protects}
In this subsection, we examine when imposing WE-fairness could help the disadvantaged group, even in the presence of a mismatch.\footnote{We remind the reader that ``helping'' and  ``harming'' is always with respect to the actual underlying utility.} A natural observation is that if the assumed and underlying utilities ``almost'' match, the results of imposing $u$-WE-fairness or $v$-WE-fairness should be roughly the same. We postpone such quantitative statements to the end of this subsection and first address easy-to-check general conditions guaranteeing that fairness is not harmful. Theorem~\ref{th_robust} below shows that only a  lenient condition is required to assure that imposing WE-fairness benefits the disadvantaged group. 
\begin{theorem}\label{th_robust}
 If $v$ and $u$ agree on which group is disadvantaged, then the $u$-WE classifier weakly increases $v$-welfare of the disadvantaged group, i.e.,
{
\thinmuskip=0mu
\medmuskip=0mu plus 0mu minus 0mu
\thickmuskip=0mu plus 0mu
\begin{equation*}\small 
\begin{cases}
     W_{ u,\UC}(a)< W_{ u,\UC}(1-a)  \\
     W_{v,\UC}(a)< W_{v,\UC}(1-a) 
\end{cases}
\Longrightarrow W_{v,\UC}(a)\leq W_{v,c_{\GE( u)}^*}(a).  
\end{equation*}}
\end{theorem}
At first glance, this theorem may look rather intuitive. However, the claim is non-trivial even if there is no mismatch, i.e., when $v\equiv u$. To see this, recall that the WE-fairness constraint is imposed on the bank: a self-interested party, which is going to find the revenue-optimal way to satisfy the constraint. One  possible way to achieve welfare equality is to give no loans to both protected groups thus harming both of them. As we show in~{\ifnum\Includeappendix=1{Section~\ref{sect_unawareness_harms} in the appendix}\else{the appendix}\fi}, such an undesired behavior is possible when Unawareness is imposed. However, it never happens under the WE-fairness; we use this inherent property of WE-fairness as a tool to prove  Theorem~\ref{th_robust}. Momentarily, let us assume that there is no mismatch, i.e., that the assumed utility and the underlying one are exactly the same. In such a case, the following Lemma~\ref{lm_no_mismatch} suggests that not only imposing WE-fairness always improves the welfare of the disadvantaged group, but also that every \textit{individual} in the disadvantaged group is weakly better off.
\begin{lemma}[Matching utilities]\label{lm_no_mismatch}
The bank-optimal $u$-WE classifier makes the $u$-disadvantaged protected group $\{A=a\}$ weakly better off at the expense of the advantaged group. Formally,
\begin{equation*}
\thinmuskip=0mu
\medmuskip=1mu plus 1mu minus 1mu
\thickmuskip=1mu plus 1mu
\footnotesize{ W_{u,\UC}(a)< W_{u,\UC}(1-a)\Rightarrow
\begin{cases}
    W_{u,\UC}(a) \leq   W_{u,c_{\GE(u)}^*}(a)  \\
    W_{u,\UC}(1-a) \geq W_{u,c_{\GE(u)}^*}(1-a)
\end{cases}}.   
\end{equation*}
Moreover, any borrower from the $u$-disadvantaged group who receives a loan under the unconstrained classifier, receives it under the bank-optimal $u$-WE one.\footnote{While our paper is focused on group notions of fairness, we stress that this result provides stronger ``individual'' guarantees.}  Formally, for every $x\in\X$ it holds that
    \begin{equation*}
 W_{u,\UC}(a)< W_{u,\UC}(1-a)\Longrightarrow \UC(x,a)\leq   c_{\GE(u)}^*(x,a).
\end{equation*}
\end{lemma}
\begin{proofof}{Lemma~\ref{lm_no_mismatch}}
By Proposition~\ref{prop_optimal_GE}, the welfare level $w^*=W_{u,c_{\GE(u)}^*}(0)=W_{u,c_{\GE(u)}^*}(1)$ achieved by the $u$-WE classifier maximizes the revenue $\P(A=0)R_{0}^*(w)+ \P(A=1)R_{1}^*(w)$ as a function of welfare level $w$. The sub-group revenues $R_a^*(w)$, $a\in\{0,1\}$ are concave functions; thus the welfare level $w^*$ lies between their maxima. These maxima are attained at the welfare levels of the bank-optimal unconstrained classifier; therefore, the welfare level $w^*$ is between $W_{u,\UC}(a)$ and $W_{u,\UC}(1-a)$. See Figure~\ref{fig:age} for illustration.

The second part of the lemma, the individual guarantees, follow from the threshold structure of the bank-optimal constrained classifier $c_{\GE(u)}^*$,  which we identified in Proposition~\ref{prop_optimal_GE}. Since the welfare level $w^*$ is above the maximum of~$R_{a}^*(w)$ (for $\{A=a\}$ being the disadvantaged group), its super-gradient contains $\lambda_a\leq 0$; therefore, $\UC(x,a)$, which corresponds to $\lambda_a=0$, is below $   c_{\GE(u)}^*(x,a).$ 
\end{proofof}
Equipped with Lemma~\ref{lm_no_mismatch}, we are ready to prove Theorem~\ref{th_robust}.
\begin{proofof}{Theorem~\ref{th_robust}}
We apply Lemma~\ref{lm_no_mismatch} to the assumed utility function $u$. By the second part of the lemma, after imposing the ${u}$-WE constraint every borrower $x\in \X$ from the ${u}$-disadvantaged group $\{A=a\}$ who received loans under the unconstrained classifier still receives them. Namely,  
$\UC(x,a)\leq   c_{\GE({u})}^*(x,a)$ for all $x\in \X$. 
Multiplying both sides by the actual underlying utility $v(x,a,y)$, substituting $X=x$ and $Y=y$, and taking expectation, we get $W_{v,\UC}(a)< W_{v,c_{\GE({u})}^*}(a).$ 
In other words, the bank-optimal classifier $c_{\GE( {u})}^*$ for the utilities $ {u}$ assumed by the regulator improves the welfare of the group $\{A=a\}$ with respect to the underlying utilities~$v$. Since $v$ and $u$ agree on the disadvantaged group's identity, this classifier improves the disadvantaged group's underlying welfare.
\end{proofof}
In the presence of a mismatch, it is easy to see that a weak converse to Theorem~\ref{th_robust} holds, i.e., fairness always makes the disadvantaged group weakly worse off. To avoid a mismatch, one can use a simple quantitative criterion.
We say that $u$ is an $\alpha$-approximation of $v$ for some $\alpha\geq 1$ if for all $x,a$, and $y$, 
$\frac{1}{\alpha}\leq \frac{v(x,a,y)}{{u}(x,a,y)} \leq \alpha.$ 
 Theorem~\ref{th_robust} implies the following quantitative criterion on how well the regulator should approximate the underlying utilities to help the disadvantaged group.
\begin{corollary}\label{cor_alpha_approximation}
If the utility ${u}$ assumed by the regulator is an $\alpha$-approximation of the underlying utility $v$ and the gap between the welfare of the groups with respect to ${u}$ is big enough, namely, $\frac{W_{u,\UC}(0)}{W_{u,\UC}(1)} \in \left(0,\frac{1}{\alpha^2}\right)\cup\left(\alpha^2,+\infty\right)$, 
then ${u}$-WE classifier helps the $v$-disadvantaged group.
\end{corollary}

\section{Computing Bank-Optimal Welfare-Equalizing Classifiers}\label{sect_computation}
In this section, we provide evidence for the applicability of WE-fairness by developing tools for computing bank-optimal WE classifiers. Our goal is to show how the bank can use the assumed utility proposed by the regulator to compute approximately optimal classifiers. We first discuss the case where the underlying and assumed utilities match (i.e., $v\equiv u$) and the other objects are fully known (revenues and losses $\alpha_\pm(x)$, and the probability $p(x,a)$ of paying back). Later on, we relax this assumption. Due to space considerations and our desire to focus on the conceptual assets of the paper, we defer most of the analysis to the appendix, as well as an elaborated version of formal statements.

Consider the case where $u\equiv v$, and both $\alpha_\pm$ and $p$ are known. If the data is tabular, i.e., $\X$ is relatively small (say several thousand different borrower types), we can compute the bank-optimal $u$-WE-classifier $c^*_{\GE(u)}$ explicitly by standard LP-methods. For large sets of attributes, e.g., multidimensional or continuous, the size of the LP ``explodes'', and a different approach should be taken. In this case, we use the structural insights from  Proposition~\ref{prop_optimal_GE}: the  bank-optimal WE-classifier $c_{\mathrm{WE}(u)}^*$ is parameterized by the two thresholds $\lambda_a$ for $a\in \{0,1\}$; therefore, to compute it we can restrict our attention to a finite-dimensional parametric family of classifiers. Due to the large-scale nature of the problem, we shall seek efficient algorithms for computing approximately bank-optimal WE classifiers, where these approximated solutions are defined as follows.
\begin{definition}\label{def:approx}
A classifier $c$ is $(\varepsilon,\varepsilon')$ bank-optimal $u$-WE if $R(c)\geq R(c^*_{\mathrm{WE}(u)})-\varepsilon$ and $|W_{u,c}(0)-W_{u,c}(1)|\leq \varepsilon'$. 
\end{definition}
Notice that such classifiers are doubly approximated: they approximate the revenue of the (exact) bank-optimal WE classifier, and also approximately equalize the welfare of the two classes. In the appendix, we demonstrate how the bank  can efficiently find a classifier that approximates the revenue and (exactly) equalizes the welfare.  Namely, we show how one can apply the ternary search method (as in~\citet{hardt2016equality}) to find $(\varepsilon,0)$ $u$-WE-classifier in $O\left(\log^2\left(\frac{1 }{\varepsilon}\right)\right)$ run-time.

Next, we get rid of the full information assumption. Recent papers \cite{balcan2019envy,hossaindesigning} propose convex relaxations for imposing fairness constraints in settings like ours, which includes generalization bounds. However, since an extensive body of literature deals with estimating real-valued functions (ranging from linear regression to deep learning), we take a different approach. We suggest that the bank employs the assumed utility given by the regulator, and describe its performance guarantees in terms of the ``quality'' of $u$. This perspective has been adopted recently for several other ML problems \cite{medina2017revenue,lykouris2018competitive,purohit2018improving}. 

For simplicity, we assume that  $|r(x,a)|$, $u(x,a,y)$, and $v(x,a,y)$ are all upper-bounded by $1$. 
\begin{proposition} \label{prop:sample_and_estimators} 
Fix a small $\delta>0$ and assume that the bank has access to a sample of $(X,A,Y,\alpha_\pm,v)$ and to estimators ${u}$ and $\hat{r}$ such that $\E\left[ \abs{{{u}}-v}\right]\leq \eta_u$ and $\E\left[\abs{\hat{r}-r}\right]\leq \eta_r$ for small enough $\eta_u$ and $\eta_r$. Then, a  $(\varepsilon,\varepsilon)$ bank-optimal  $v$-WE classifier with 
\[
\varepsilon=2\sqrt{6\left(\frac{1}{\P(A=0)}+\frac{1}{\P(A=1)}\right)\max\{\eta_u,\eta_r\}}
\]
can be computed with probability $1-\delta$ on a sample of size 
\[
O\left(\frac{1}{\max\{\eta_u,\eta_r\}}\left(\log\frac{1}{\delta}+\log\log\frac{1}{\max\{\eta_u,\eta_r\}}\right)\right).
\]
\end{proposition}

\section{Conclusions}
Our paper draws on the economic approach to  fair classification. It  initiates the discussion of the impact that the regulator's misconception about the population characteristics may have on the protected groups' well-being. Beyond that, we believe that our WE-fairness can serve as an anchor for grounding other fairness stances.

We have prioritized clarity over generality and focused on binary classification. Nevertheless, our results are more general.
The key technical Proposition~\ref{prop_optimal_GE} 
can be extended to multiple classes (e.g., several loans with different periods and interest rates).\footnote{As in the binary case, there is one threshold $\lambda_a$ per group $a$, but now the revenue $r$ and the utilities $\overline{u}$ depend on the class. The optimal WE-classifier selects a class $c$ with  the maximal $r-\lambda_a\cdot \overline{u}$.} All other results (the mismatch analysis and the algorithms) are, essentially, corollaries of Proposition~\ref{prop_optimal_GE} and extend straightforwardly. 
Noisy utilities and revenues can be assumed for free if we interpret $u$, $v$, and $r$ in all formulas as the conditional expectations for a given triplet $(X,A,Y)$. Allowing for utilities with negative expectations also does not alter the statements.
We see considerable scope for follow-up work. One prominent direction is to understand how the ``price of fairness'' is distributed among the parties, e.g., by how much the bank’s revenue and the advantaged group’s welfare drop.

\addcontentsline{toc}{section}{\protect\numberline{}References}%

{\ifnum\Putacknowledgement=1{
\section*{Acknowledgements}

We are grateful to Lily Hu, Nikita Kalinin, Alexander Nesterov, Margarita Niyazova, and Ivan Susin for insightful discussions and pointers to the relevant literature. We also thank the anonymous reviewers for their helpful comments and clarifications and  Lillian Bluestein for proofreading.

The work of M. Tennenholtz is funded by the European Research Council (ERC) under the European Union's Horizon 2020 research and innovation programme (grant agreement n$\degree$  740435). F. Sandomirskiy is partially supported by the Lady Davis Foundation, by Grant 19-01-00762 of the Russian Foundation for Basic Research, and by Basic Research Program of the National Research University Higher School of Economics. This work was done while O. Ben-Porat and F. Sandomirskiy were at the Technion---Israel Institute of Science and were partially funded by the European Union's Horizon 2020 research and innovation programme (grant agreement n$\degree$  740435).
}\fi}

{\small

\begin{thebibliography}{42}
\providecommand{\natexlab}[1]{#1}
\providecommand{\url}[1]{\texttt{#1}}
\providecommand{\urlprefix}{URL }
\expandafter\ifx\csname urlstyle\endcsname\relax
  \providecommand{\doi}[1]{doi:\discretionary{}{}{}#1}\else
  \providecommand{\doi}{doi:\discretionary{}{}{}\begingroup
  \urlstyle{rm}\Url}\fi

\bibitem[{Agarwal et~al.(2018)Agarwal, Beygelzimer, Dud{\'\i}k, Langford, and
  Wallach}]{agarwal2018reductions}
Agarwal, A.; Beygelzimer, A.; Dud{\'\i}k, M.; Langford, J.; and Wallach, H.
  2018.
\newblock A reductions approach to fair classification.
\newblock \emph{arXiv preprint arXiv:1803.02453} .

\bibitem[{Balcan et~al.(2019)Balcan, Dick, Noothigattu, and
  Procaccia}]{balcan2019envy}
Balcan, M.-F.~F.; Dick, T.; Noothigattu, R.; and Procaccia, A.~D. 2019.
\newblock Envy-free classification.
\newblock In \emph{Advances in Neural Information Processing Systems},
  1238--1248.

\bibitem[{Benade et~al.(2018)Benade, Kazachkov, Procaccia, and
  Psomas}]{benade2018make}
Benade, G.; Kazachkov, A.~M.; Procaccia, A.~D.; and Psomas, C.-A. 2018.
\newblock How to make envy vanish over time.
\newblock In \emph{Proceedings of the 2018 ACM Conference on Economics and
  Computation}, 593--610. ACM.

\bibitem[{Brandt et~al.(2016)Brandt, Conitzer, Endriss, Lang, and
  Procaccia}]{brandt2016handbook}
Brandt, F.; Conitzer, V.; Endriss, U.; Lang, J.; and Procaccia, A.~D. 2016.
\newblock \emph{Handbook of computational social choice}.
\newblock Cambridge University Press.

\bibitem[{Caragiannis et~al.(2019)Caragiannis, Kurokawa, Moulin, Procaccia,
  Shah, and Wang}]{caragiannis2019unreasonable}
Caragiannis, I.; Kurokawa, D.; Moulin, H.; Procaccia, A.~D.; Shah, N.; and
  Wang, J. 2019.
\newblock The unreasonable fairness of maximum Nash welfare.
\newblock \emph{ACM Transactions on Economics and Computation (TEAC)} .

\bibitem[{Cohler et~al.(2011)Cohler, Lai, Parkes, and
  Procaccia}]{cohler2011optimal}
Cohler, Y.~J.; Lai, J.~K.; Parkes, D.~C.; and Procaccia, A.~D. 2011.
\newblock Optimal envy-free cake cutting.
\newblock In \emph{Twenty-Fifth AAAI Conference on Artificial Intelligence}.

\bibitem[{Collinson(14 January 2017)}]{collinson2017}
Collinson, P. 14 January 2017.
\newblock How an EU gender equality ruling widened inequality.
\newblock \emph{The Guardian} .

\bibitem[{Corbett-Davies and Goel(2018)}]{corbett2018measure}
Corbett-Davies, S.; and Goel, S. 2018.
\newblock The measure and mismeasure of fairness: A critical review of fair
  machine learning.
\newblock \emph{arXiv preprint arXiv:1808.00023} .

\bibitem[{Corbett-Davies et~al.(2017)Corbett-Davies, Pierson, Feller, Goel, and
  Huq}]{corbett2017algorithmic}
Corbett-Davies, S.; Pierson, E.; Feller, A.; Goel, S.; and Huq, A. 2017.
\newblock Algorithmic decision making and the cost of fairness.
\newblock In \emph{Proceedings of the 23rd ACM SIGKDD International Conference
  on Knowledge Discovery and Data Mining}, 797--806. ACM.

\bibitem[{Doleac and Hansen(2016)}]{doleac2016does}
Doleac, J.~L.; and Hansen, B. 2016.
\newblock Does “ban the box” help or hurt low-skilled workers? Statistical
  discrimination and employment outcomes when criminal histories are hidden.
\newblock Technical report, National Bureau of Economic Research.

\bibitem[{Dutta and Ray(1989)}]{dutta1989concept}
Dutta, B.; and Ray, D. 1989.
\newblock A concept of egalitarianism under participation constraints.
\newblock \emph{Econometrica: Journal of the Econometric Society} 615--635.

\bibitem[{Dwork et~al.(2012)Dwork, Hardt, Pitassi, Reingold, and
  Zemel}]{dwork2012fairness}
Dwork, C.; Hardt, M.; Pitassi, T.; Reingold, O.; and Zemel, R. 2012.
\newblock Fairness through awareness.
\newblock In \emph{Proceedings of the 3rd Innovations in Theoretical Computer
  Science Conference}, 214--226. ACM.

\bibitem[{Dwork et~al.(2018)Dwork, Immorlica, Kalai, and
  Leiserson}]{dwork2018decoupled}
Dwork, C.; Immorlica, N.; Kalai, A.~T.; and Leiserson, M. 2018.
\newblock Decoupled classifiers for group-fair and efficient machine learning.
\newblock In \emph{Conference on Fairness, Accountability and Transparency},
  119--133.

\bibitem[{Elzayn and Fish(2020)}]{ElzaynFat2020}
Elzayn, H.; and Fish, B. 2020.
\newblock The effects of competition and regulation on error inequality in
  data-driven markets.
\newblock In \emph{Proceedings of the 2020 Conference on Fairness,
  Accountability, and Transparency}, FAT* '20, 669–679. Association for
  Computing Machinery.

\bibitem[{Finocchiaro et~al.(2020)Finocchiaro, Maio, Monachou, Patro, Raghavan,
  Stoica, and Tsirtsis}]{finocchiarofairness}
Finocchiaro, J.; Maio, R.; Monachou, F.; Patro, G.~K.; Raghavan, M.; Stoica,
  A.-A.; and Tsirtsis, S. 2020.
\newblock Fairness and Discrimination in Mechanism Design and Machine Learning.
\newblock In \emph{AI for Social Good Workshop}.

\bibitem[{Gal et~al.(2016)Gal, Mash, Procaccia, and Zick}]{gal2016fairest}
Gal, Y.~K.; Mash, M.; Procaccia, A.~D.; and Zick, Y. 2016.
\newblock Which is the fairest (rent division) of them all?
\newblock In \emph{Proceedings of the 2016 ACM Conference on Economics and
  Computation}, 67--84. ACM.

\bibitem[{G{\"o}lz, Kahng, and Procaccia(2019)}]{golz2019paradoxes}
G{\"o}lz, P.; Kahng, A.; and Procaccia, A.~D. 2019.
\newblock Paradoxes in Fair Machine Learning.
\newblock In \emph{Advances in Neural Information Processing Systems},
  8340--8350.

\bibitem[{Guruswami et~al.(2005)Guruswami, Hartline, Karlin, Kempe, Kenyon,
  McSherry, and McSherry}]{guruswami2005profit}
Guruswami, V.; Hartline, J.~D.; Karlin, A.~R.; Kempe, D.; Kenyon, C.; McSherry,
  F.; and McSherry, F. 2005.
\newblock On profit-maximizing envy-free pricing.
\newblock In \emph{ACM-SIAM Symposium on Discrete Algorithms}.

\bibitem[{Hardt et~al.(2016)Hardt, Price, Srebro et~al.}]{hardt2016equality}
Hardt, M.; Price, E.; Srebro, N.; et~al. 2016.
\newblock Equality of opportunity in supervised learning.
\newblock In \emph{Advances in Neural Information Processing Systems},
  3315--3323.

\bibitem[{Heidari et~al.(2018)Heidari, Ferrari, Gummadi, and
  Krause}]{heidari2018fairness}
Heidari, H.; Ferrari, C.; Gummadi, K.; and Krause, A. 2018.
\newblock Fairness behind a veil of ignorance: A welfare analysis for automated
  decision making.
\newblock In \emph{Advances in Neural Information Processing Systems},
  1265--1276.

\bibitem[{Heidari, Gummadi, and Krause(2019)}]{heidarifat19}
Heidari, H.; Gummadi, M. L. K.~P.; and Krause, A. 2019.
\newblock Moral framework for understanding fair ML through economic models of
  equality of opportunity.
\newblock In \emph{FAT* 19}, 181--190.

\bibitem[{Hossain, Mladenovic, and Shah(2020)}]{hossaindesigning}
Hossain, S.; Mladenovic, A.; and Shah, N. 2020.
\newblock Designing fairly fair classifiers via economic fairness notions.
\newblock In \emph{Proceedings of The Web Conference 2020}, 1559–1569.

\bibitem[{Hu and Chen(2020)}]{HuFat2020}
Hu, L.; and Chen, Y. 2020.
\newblock Fair classification and social welfare.
\newblock In \emph{Proceedings of the 2020 Conference on Fairness,
  Accountability, and Transparency}, 535–545. Association for Computing
  Machinery.
\newblock ISBN 9781450369367.

\bibitem[{John(1985)}]{john1985course}
John, B.~C. 1985.
\newblock \emph{A Course in functional analysis}.
\newblock New York: Springer-Verlag.

\bibitem[{Kalai and Smorodinsky(1975)}]{kalai1975other}
Kalai, E.; and Smorodinsky, M. 1975.
\newblock Other solutions to Nash’s bargaining problem.
\newblock \emph{Econometrica} 43(3): 513--518.

\bibitem[{Larson et~al.(2016)Larson, Mattu, Kirchner, and
  Angwin}]{larson2016we}
Larson, J.; Mattu, S.; Kirchner, L.; and Angwin, J. 2016.
\newblock How we analyzed the COMPAS recidivism algorithm.
\newblock \emph{ProPublica (5 2016)} 9.

\bibitem[{Li and Xue(2013)}]{li2013egalitarian}
Li, J.; and Xue, J. 2013.
\newblock Egalitarian division under Leontief preferences.
\newblock \emph{Economic Theory} 54(3): 597--622.

\bibitem[{Liu et~al.(2018)Liu, Dean, Rolf, Simchowitz, and
  Hardt}]{liu2018delayed}
Liu, L.; Dean, S.; Rolf, E.; Simchowitz, M.; and Hardt, M. 2018.
\newblock Delayed impact of fair machine learning.
\newblock In \emph{International Conference on Machine Learning}, 3156--3164.

\bibitem[{Lykouris and Vassilvitskii(2018)}]{lykouris2018competitive}
Lykouris, T.; and Vassilvitskii, S. 2018.
\newblock Competitive caching with machine learned advice.
\newblock \emph{arXiv preprint arXiv:1802.05399} .

\bibitem[{Medina and Vassilvitskii(2017)}]{medina2017revenue}
Medina, A.~M.; and Vassilvitskii, S. 2017.
\newblock Revenue optimization with approximate bid predictions.
\newblock In \emph{Proceedings of the 31st International Conference on Neural
  Information Processing Systems}, 1856--1864.

\bibitem[{O'Neill(2016)}]{o2016weapons}
O'Neill, C. 2016.
\newblock \emph{Weapons of math destruction: How big data increases inequality
  and threatens democracy}.
\newblock Crown Publishers.

\bibitem[{Parker and Wang(2013)}]{parker2013modern}
Parker, K.; and Wang, W. 2013.
\newblock Modern parenthood.
\newblock \emph{Pew Research Center’s Social \& Demographic Trends Project}
  14.

\bibitem[{Pazner and Schmeidler(1978)}]{pazner1978egalitarian}
Pazner, E.~A.; and Schmeidler, D. 1978.
\newblock Egalitarian equivalent allocations: A new concept of economic equity.
\newblock \emph{The Quarterly Journal of Economics} 92(4): 671--687.

\bibitem[{Plaut and Roughgarden(2018)}]{plaut2018almost}
Plaut, B.; and Roughgarden, T. 2018.
\newblock Almost envy-freeness with general valuations.
\newblock In \emph{Proceedings of the Twenty-Ninth Annual ACM-SIAM Symposium on
  Discrete Algorithms}, 2584--2603. Society for Industrial and Applied
  Mathematics.

\bibitem[{Purohit, Svitkina, and Kumar(2018)}]{purohit2018improving}
Purohit, M.; Svitkina, Z.; and Kumar, R. 2018.
\newblock Improving online algorithms via ml predictions.
\newblock In \emph{Advances in Neural Information Processing Systems},
  9661--9670.

\bibitem[{Rambachan et~al.(2020{\natexlab{a}})Rambachan, Kleinberg, Ludwig, and
  Mullainathan}]{rambachan2020economic}
Rambachan, A.; Kleinberg, J.; Ludwig, J.; and Mullainathan, S.
  2020{\natexlab{a}}.
\newblock An Economic Perspective on Algorithmic Fairness.
\newblock In \emph{AEA Papers and Proceedings}, volume 110, 91--95.

\bibitem[{Rambachan et~al.(2020{\natexlab{b}})Rambachan, Kleinberg,
  Mullainathan, and Ludwig}]{rambachan2020economicA}
Rambachan, A.; Kleinberg, J.; Mullainathan, S.; and Ludwig, J.
  2020{\natexlab{b}}.
\newblock An economic approach to regulating algorithms.
\newblock Technical report, National Bureau of Economic Research.

\bibitem[{Rawls(2009)}]{rawls2009theory}
Rawls, J. 2009.
\newblock \emph{A theory of justice}.
\newblock Harvard university press.

\bibitem[{Rockafellar(2015)}]{rockafellar2015convex}
Rockafellar, R.~T. 2015.
\newblock \emph{Convex analysis}.
\newblock Princeton university press.

\bibitem[{Roemer(1986)}]{roemer1986equality}
Roemer, J.~E. 1986.
\newblock Equality of resources implies equality of welfare.
\newblock \emph{The Quarterly Journal of Economics} 101(4): 751--784.

\bibitem[{Ustun, Liu, and Parkes(2019)}]{ustun2019fairness}
Ustun, B.; Liu, Y.; and Parkes, D. 2019.
\newblock Fairness without harm: Decoupled classifiers with preference
  guarantees.
\newblock In \emph{International Conference on Machine Learning}, 6373--6382.

\bibitem[{Zafar et~al.(2017)Zafar, Valera, Rodriguez, Gummadi, and
  Weller}]{zafar2017parity}
Zafar, M.~B.; Valera, I.; Rodriguez, M.; Gummadi, K.; and Weller, A. 2017.
\newblock From parity to preference-based notions of fairness in
  classification.
\newblock In \emph{Advances in Neural Information Processing Systems},
  229--239.

\end{thebibliography}

}

{\ifnum\Includeappendix=1{ 
\newpage
\onecolumn
\appendix

\section{\fed Threshold Structure of the Bank-optimal WE-classifier}

\begin{proofof}{Proposition \ref{prop_optimal_GE}}
The revenue maximization under the WE constraint can be split into two subsequent maximization problems:
$$\max_{c\in \GE(u)} R(c)=\max_{w\in\R}\max_{{\begin{array}{c} \scriptstyle{c: \, \X\times A\to[0,1]}\\ \scriptstyle{W_{u,c}(0)=W_{u,c}(1)=w}  \end{array}}} \P(A=0)R_{0}(c)+\P(A=1)R_{1}(c).$$
First, the bank finds the revenue-maximizing classifier $c_{a,w}^*: \X\to[0,1]$ that maximizes the revenue $R_a(c)=\E[r(X,A)\cdot c(X)\mid A=a]$ in the subgroup $\{A=a\}$ given some welfare-level $W_{u,c}(a)=w$ {\fed (we refer to these classifiers as optimal marginal classifiers).} Then, the bank finds the {\fed bank-optimal} level $w^*$ of $w$ by maximizing the total revenue $\P(A=0)R_{0}^*(w)+\P(A=1)R_{1}^*(w)$, where as in the statement of the proposition, $R_{a}^*(w)$ denotes $R_a(c_{a,w}^*)$.

Thus, the {\fed bank-optimal} WE classifier $c^*_{\GE(u)}(x,a)$ equals $c_{a,w^*}^*(x)$ provided that the optimization problems for $c_{a,w}^*$ and $w^*$ have a solution. The threshold representation in Equation \eqref{eq_optGE_classifier} will follow from a similar representation {\fed for the optimal marginal classifiers} $c_{a,w}^*$. Existence and the threshold representation of $c_{a,w^*}^*$ is the subject of the next two subsections.

\subsection{Existence of \texorpdfstring{$c_{a,w}^*$ and $w^*$, concavity of $R_{a}^*(w)$}{{\fed bank-optimal} classifier}}
For a given welfare-level $w$ of a group $\{A=a\}$, the set of feasible {\fed marginal} classifiers $F(w)=\{c: \X\to[0,1]\mid W_{u,c}(a)=w\}$ is non-empty if and only if  $w\in [0,\E[u\mid A=a]]$.
Indeed, for any $c$ we have $W_{u,c}(a)\in [0,\E[u\mid A=a]]$; therefore, the set $F(w)$ is empty outside of this interval. For any $w$ inside, the constant classifier $\frac{w}{\E[u\mid A=a]}\in F(w)$ and thus $F(w)$ is non-empty. 
Therefore, for $w\in [0,\E[u_+\mid A=a]]$,  $R_{a}^*(w)=\sup_{c\in F(w)} R_a(c)$ is finite; outside this interval we assume $R_{a}^*(w)=-\infty$. Let us show that $R_{a}^*(w)$ is concave and continuous  and that supremum is, in fact, maximum, i.e., that the {\fed bank-optimal} (marginal) classifier $c_{a,w}^*$ exists.

For any $c'\in F(w') $ and $c''\in F(w'')$ the convex combination $c=\beta c'+(1-\beta)c''$ for $\beta\in[0,1]$ belongs to $F(w)$, where $w=\beta w'+(1-\beta )w''$. 
Therefore,
$R_{a}^*(w)\geq \beta R_a(c')+(1-\beta)R_a(c'')$. Taking supremum over $c'\in F(w')$ and $c''\in F(w'')$ we obtain $R_{a}^*(w)\geq \beta R_{a}^*(w')+(1-\beta)R_{a}^*(w'')$. Thus $R_{a}^*(w)$ is concave in $w$.

Next, we prove that the maximum is attained. Consider $c^{(n)}\in F(w), n=1,2,...$, a sequence of classifiers such that $R_{a}(c^{(n)})\to R^*_{a}(w)$. This sequence is a bounded set in the Hilbert space $L^2(\X,\P\vert_{(X\mid A=a)})$; thus, by the Banach-Alaoglu theorem  \cite[Section 5]{john1985course}, this sequence contains a weakly-convergent subsequence $c^{(n_k)}\to c\in F(w)$. Since $R_a(c^{n_k})$ is a scalar product of $c$ and $r(\,\cdot\, ,a)$ in $L^2(\X,\P\vert_{(X\mid A=a)})$, we get $R_a(c)=\lim_{k\to\infty} R_a(c^{n_k})=R_{a}^*(w)$. Thus the maximum is attained.

A similar argument proves continuity of $R_{a}^*(w)$. By concavity, continuity on $[0,\E[u\mid A=a]]$ follows from upper-semi-continuity:  $R_{a}^*(w)\geq \lim_{n\to\infty} R^*_{a}(w^{(n)})$ for any sequence $w^{(n)}\to w, \ n\to\infty.$ Consider the sequence of {\fed bank-optimal} classifiers $c^*_{a,w^{(n)}}, \ n=1,2,\dots$. There is a weakly-convergent subsequence $c^*_{a,w^{(n_k)}}$ with some $c$ as the weak limit. We have $R_a(c)=\lim_{k\to\infty} R^*_{a}(w^{(n_k)})$ and similarly  $W_{u,c}(a)=\lim_{k\to \infty} w^{(n_k)}=w$. Thus $c\in F(w)$ and $R_{a}^*(w)\geq \lim_{n\to\infty} R^*_{a}(w^{(n)})$.

The existence of the {\fed bank-optimal} $w^*$, where
\[
w^* \in \arg\max_w \{ \P(A=0)R_{0}^*(w)+\P(A=0)R_{1}^*(w) \}
\]
follows from the continuity of $R^*_{a}(w)$ for  $w\in [0,\E[u\mid A=a]]$. 

\subsection{Threshold structure of the {\fed bank-optimal} marginal classifiers \texorpdfstring{$c_{a,w^*}^*$}{}}
By concavity and continuity of $R^*_{a}(w)$, the super-gradient at $w^*$ is non-empty. Pick an element $\lambda_a$ from it.

By the definition of super-gradient, for any $w\in\R$ it holds that 
$R^*_{a}(w)\leq R^*_{a}(w^*)+\lambda_a(w-w^*)$. Equivalently, for the {\fed bank-optimal} classifier $c_{a,w^*}^*\in F(w^*)$ and an arbitrary classifier $c : \  \X\to[0,1]$ we have
$$ R_{a}(c)-\lambda_a W_{u,c}(a)\leq R_{a}(c_{a,w^*}^*)-\lambda_a W_{u,c_{a,w^*}^*}(a).$$
In other words, the {\fed bank-optimal} classifier $c_{a,w^*}^*$ maximizes $R_{a}(c)-\lambda_a W_{u,c}(a)$ over all 
$c : \  \X\to[0,1]$. The converse is also true: any maximizer $c$ gives a {\fed bank-optimal} classifier $c_{a,w^*}$ provided that it belongs to $F(w^*)$.

Since
$$R_{a}(c)-\lambda_a W_{u,c}(a)= \E\left[\left(r(X,A)-\lambda_a\overline{u}(X,A) \right)c(X,A) \mid A=a   \right],$$
the maximizer equals $1$ when $r(x,a)-\lambda_a\overline{u}(x,a)>0$ and equals $0$ if the inequality has the strict opposite sign. The condition $c\in F(w^*)$
imposes the constraint on otherwise arbitrary values of $c(x,a)=\tau_a(x)$ for $x$ with $r(x,a)-\lambda_a\overline{u}(x,a)=0$.
\end{proofof}

\subsection{{\fed Example: the bank-optimal classifiers for DP and EO}}\label{subsect_ex_threshold_DP_EO}

Proposition~\ref{prop_optimal_GE} immediately provides the {\fed bank-optimal} classifiers under DP and EO: for the former, it is given by the ``additive perturbation'' of the optimal unconstrained classifier, while the latter is obtained by the ``multiplicative perturbation.'' 

 More precisely, {\fed for the optimal classifier $c_\textnormal{parity}^*$ satisfying DP we recover the result of \citet{corbett2017algorithmic}:} there exist group-dependent constants $\lambda_a\in \R$, $a\in\{0,1\}$ such that
\[
c_\textnormal{parity}(x,a)=
\begin{cases}
1     &  p(x,a)> t(x,a)+\frac{\lambda_a}{\alpha_+(x)+\alpha_-(x)} \\
0     &  p(x,a)< t(x,a)+\frac{\lambda_a}{\alpha_+(x)+\alpha_-(x)}
\end{cases},
\]
{\fed where $p(x,a)$ is the probability that a borrower with attributes $X=x$ and $A=a$ pays back, and the threshold $t(x)$ is defined by~\eqref{eq_t(x)}.}

 For the {\fed bank-optimal} EO classifier $c_\textnormal{opportunity}^*$, there exist $\lambda_a\in \R$, $a\in\{0,1\}$ such that
{
\thinmuskip=0mu
\medmuskip=1mu plus 1mu minus 1mu
\thickmuskip=2mu plus 2mu
\[
c_\textnormal{opportunity}^*(x,a)=
\begin{cases}
1     &  p(x,a)\left(1-\frac{\lambda_a}{\alpha_+(x)+\alpha_-(x)}\right)> t(x,a)\\
0     &  p(x,a)\left(1-\frac{\lambda_a}{\alpha_+(x)+\alpha_-(x)}\right)< t(x,a)
\end{cases}.
\]
}


\section{Unconstrained bank-optimal classifier is unfair}
{\fed Here we show that if the bank is free to pick any classifier (i.e., the regulator imposes no constraints), the  bank-optimal \textit{unconstrained} classifier (formally defined in Subsection~\ref{subsect_optimal_GE}) can discriminate against one of the groups even if \textit{the groups are of equal size and contain the same fraction of good borrowers}. {\fed We illustrate this claim by the following example.}}
\begin{example}[Bank-optimal unconstrained classifier is unfair]\label{ex_unconstrained}
Let $A$ and $X$ be binary, and let all four combinations {\fed have the probability $\frac{1}{4}$.} The probability $p(x,a)$ of being a good borrower is given by the matrix
$$\small \begin{array}{c|cc}
        & x=0 & x=1 \\
        \hline
a=0 & 0.4 & 0.6\\
a=1 & 0 & 1
\end{array}.$$
In the group $\{A=0\}$, the attribute $X$ poorly separates good and bad borrowers, while it is a perfect predictor of creditworthiness for $\{A=1\}$. If losses from a defaulting client are equal to the revenue from two borrowers paying back, e.g. $\alpha_-(x) = 2 \cdot \alpha_+(x)$, then {\fed the revenue-maximizing classifier has the following form: the bank gives the loans to those borrowers that have the probability of paying back above the threshold $t=\frac{\alpha_-(x)}{\alpha_-(x)+\alpha_+(x)}\equiv\frac{2}{3}\simeq 0.66$.}
 Thus, the bank gives loans to applicants with $(X,A)=(1,1)$ only. As a result, \emph{no loans are given to the group $\{A=0\}$}, although the prior distribution shows that in total $1/2$ of borrowers from each of the groups $\{A=0\}$ and $\{A=1\}$ are good. {\fed This contradicts the intuitive understanding of fairness. Moreover, \emph{for any utility function $u$,} the welfare of the group  $\{A=0\}$ is below the welfare of the group $\{A=1\}$, i.e., the optimal unconstrained classifier is not  $u$-Welfare-Equalizing for any $u$.}
\end{example}

\section{{\fed Unawareness Can Harm Both Groups and the Bank}}\label{sect_unawareness_harms}


\emph{Unawareness} is perhaps the most intuitive fairness criterion. A classifier $c$ satisfies unawareness if $c(x,0)=c(x,1)$ for all $x\in \X$. Informally, to deliver a ``fair'' outcome to both groups $A=0$ and $A=1$, the classifier $c$ must ignore the protected attribute $A$. 

This natural idea is not innocuous: $X$ and $A$ can be dependent and thus information contained in $X$ can be used as a proxy for $A$, making the disadvantaged group even worse off; see~\citet{doleac2016does}. Despite its flaws, unawareness is promoted by layers in the labor market, insurance market\footnote{By the decision of the Court of Justice of the European Union, insurance premiums must be determined in a gender-blind way starting from December 2012. This initiative eventually increased the gap between premiums paid by females and males~\cite{collinson2017}.}, and also by GDPR.\footnote{The General Data Protection Regulation is a law adopted by the European Union in 2016 that regulates data protection and privacy.}

The {\fed bank-optimal} unaware classifier $c_\text{unaware}^*$ can be constructed by the same logic as the optimal unconstrained one; see Subsection~\ref{subsect_GE} and \citet{corbett2017algorithmic}: an applicant receives money if the probability $\P(Y=1\mid X=x)$ of being a good borrower given the observed attribute $X=x$ is above the threshold $t(x)=\frac{\alpha_-(x)}{\alpha_+(x)+\alpha_-(x)}$. Formally, $c_\text{unaware}^*(x,a)=c_\text{unaware}^*(x)=1$ for borrowers $x$ such that $\overline{p}(x)=\P(A=0)p(x,0)+\P(A=1)p(x,1)\geq t(x)$ and $c_\text{unaware}^*(x)=0$ if $\overline{p}(x)< t(x)$.

We now show that {Unawareness} can make all three parties strictly worse off: both groups and the bank, \textit{regardless} of the underlying utility function $v$. The following example shows  this phenomenon when information losses caused by Unawareness are significant: the interpretation  of the non-protected attribute $X$ depends on the protected\footnote{For example, having children correlates with spending more time at work for men and has negative correlation for women~\cite{parker2013modern}.} attribute $A$ and without knowing $A$ the classifier cannot achieve good separation of good and bad borrowers; thus giving loans becomes too risky. 

\begin{example}
 Suppose  that $\X=\{0,1\}$ and all the four combinations of attributes $(x,a)$ have an equal probability of $\frac{1}{4}$. The fraction $p(x,a)$ of good borrowers is given by the following table
$$\begin{array} {c|cc}
        & x=0 & x=1 \\
\hline
a=0 & 2/3 & 1/3\\
a=1 & 1/3 & 2/3
\end{array}. $$
Notice that the fraction of good and bad borrowers is the same in both groups. Further, $X=1$ is a positive signal about the quality of a borrower in the group $\{A=0\}$ and a negative one for $\{A=1\}$. 

The {\fed bank-optimal} unconstrained classifier $\UC$ gives loans to agents with $p(x,a)\geq t(x)$; thus, if the threshold $t(x)$ is between $\frac{1}{3}$ and $\frac{2}{3}$ (for example, $t(x)=\frac{3}{5}$ if the revenue $\alpha_+$ from giving money to a good borrower equals $\frac{2}{3}$ of the losses $\alpha_-$ from a bad one), then agents with $(X,A)$ equal to $(0,0)$ or $(1,1)$ get loans under $\UC$. So one half of the members in each group receives loans and the bank gets a positive revenue.

After imposing {Unawareness}, the {\fed bank-optimal} classifier $c_\text{unaware}^*$ compares the threshold $t(x)$ with the average fraction of good borrowers for a given $X=x$, namely $\overline{p}(x)=\P(A=0)p(x,0)+\P(A=1)p(x,1)$. In our example, $\overline{p}(x)= \frac{1}{2}$ for every $x$, and thus no loans are given for $t(x)>1/2$. Thus, if $t(x)\in(1/2, 2/3)$ for all $x\in \X$, Unawareness pushes the welfare of both groups to zero for any underlying utility $v$ as well as the bank's revenue.
\end{example}

\section{Proof of Corollary~\ref{cor_alpha_approximation}}
\begin{proofof}{Corollary~\ref{cor_alpha_approximation}}
By Theorem~\ref{th_robust}, it is enough to show that, ${u}$ and $v$ agree on which group is disadvantaged. Without loss of generality, the group $\{A=0\}$ is $u$-disadvantaged. We shall show that $\{A=0\}$ is also $v$-disadvantaged. We get the following chain of inequalities: 
$$W_{v,\UC}(0)\leq \alpha W_{ u,\UC}(0)< \frac{1}{\alpha}W_{u,\UC}(1)\leq W_{v,\UC}(1).$$
The first and the last steps follow from the definition of $\alpha$-approximation, and the second step from the condition in the corollary. Consequently, 
$\{A=0\}$ is $v$-disadvantaged, which completes the proof.
\end{proofof}

\section{Algorithms for  Approximately Bank-optimal WE-Classifiers}\label{sect_computing_WE_appendix}

{\fed This section is devoted to finding approximately bank-optimal WE-classifiers. It contains a formal version of the results discussed in Section~\ref{sect_computation} together with their proofs.}

We first assume complete information about the problem {(\fed the utilities $u\equiv v$,} revenues/losses $\alpha_\pm$, and the probability distribution $\mathbb P$ are known), and then relax this assumption. 

{\fed Throughout this section we assume that the {\fed bank-optimal} unconstrained classifier provides non-zero utility to both groups. This positivity condition allows us to avoid degenerate cases. }

\subsection{Complete information}
{\fed In this subsection we assume that the underlying utility $v$ coincides with the utility $u$ assumed by the regulator. Furthermore, these utilities are known to the bank as well as} the distribution $\P$ and functions  $\alpha_+,\alpha_-$. In particular, the bank can compute $r(x,a)=\E[\alpha_+ Y- \alpha_-(1-Y)\mid X=x, A=a]$   and $\overline{u}(x,a)=\E[u(X,A,Y)\mid X=x, A=a]$, namely, the average revenue and the average utility for a given $(x,a)$ pair.
{\fed
\subsubsection{LP-based approach}  As we mentioned in Section~\ref{sect_computation},} the {\fed bank-optimal} $u$-WE-classifier is the solution to the linear program (LP) given in  Figure~\ref{fig:LP finite}. It maximizes the revenue under the constraint of equal welfare, and hence for a medium size set of attributes $\X$ (say several thousands) we can compute $c^*_{\GE(u)}$ explicitly by standard LP-methods.

\begin{figure}[t]
\begin{align}
&\text{maximize} \displaystyle\sum_{x,a} \P(X=x,A=a)\cdot r(x,a)\cdot c(x,a) \nonumber \\
&\text{subject to} \scriptsize \sum_{x} \P(X=x\mid A=0) \cdot \overline{u}(x,0)\cdot c(x,0)  \nonumber\\
&\qquad \qquad \quad = \sum_{x}\P(X=x\mid A=1) \cdot \overline{u}(x,1)\cdot c(x,1)  \nonumber \\
& \quad \qquad \qquad 0\leq c(x,a)\leq 1, x\in \X,a\in A  \nonumber
\end{align}
\caption{Linear programming formulation. The objective function of the LP is $R(c)$, the bank's revenue, which is maximized over the class of functions $\{\X \times A \rightarrow [0,1]\}$. The first constraint is welfare-equalizing (see Equation \eqref{eq: u-WE}); the welfare of the group $\{A=0\}$ should be equal to the welfare of  $\{A=1\}$. The remaining constraints are range constraints, assuring that the classifier is in segment $[0,1]$.\label{fig:LP finite} }
\end{figure}
\subsubsection{Threshold-based approach}
{\fed For large sets of attributes $\X$, e.g., multidimensional or continuous, solving the LP in Figure~\ref{fig:LP finite} is no longer feasible in practice. In this case, we harness the threshold structure of the bank-optimal classifier (see Proposition~\ref{prop_optimal_GE}) and use binary search over the thresholds $\lambda_a$ to compute approximately optimal solutions. The approximation notion is captured by Definition~\ref{def:approx}.}


{\fed The next proposition shows that the bank can efficiently compute a classifier that approximates the optimal revenue and achieves an exact equality of welfare. To simplify the presentation, in this subsection we assume that the group $\{A=0\}$ is disadvantaged.}
\begin{proposition}\label{prop_omniscient_algorithm}
Assume that $\alpha_+, \alpha_-$, and $u$ are known.  Further, assume that the bank has access to an oracle that computes expectations w.r.t. $\P$ in constant time. Then, for any positive $\varepsilon$ such that
\begin{equation}\label{eq_varepsilon}
 \varepsilon<\frac{W_{u,\UC}(0)}{6\max_{x,a} \overline{u}(x,a)}\min\left\{\frac{\P(A=0)}{\P(A=1)},\frac{\P(A=1)}{\P(A=0)}\right\} 
\end{equation}
an  ($\varepsilon\cdot \max_{x,a} |r(x,a)|$,\ 0) {\fed bank-optimal} $u$-WE-classifier can be computed in $O\left(\log^2\left(\frac{1 }{\varepsilon}\right)\right)$, uniformly over all other parameters.
\end{proposition}
{\fed As we show below, the assumption on the expectation oracle can be relaxed by sampling, which leads to an approximate  welfare equality.}
\begin{proofof}{Proposition \ref{prop_omniscient_algorithm}}

{\fed The high-level structure of the proof is as follows.} Similarly to the proof of Proposition~\ref{prop_optimal_GE}, we represent the optimization problem as finding an {\fed approximately bank-optimal} marginal classifier for a given welfare level $w$ in each subgroup, and then determining the revenue-maximizing $w$.
For a given welfare level $w$, the revenue-maximizing classifiers and the revenue $R^*(w)$ itself can be approximated by the binary search over thresholds $\lambda_a$. Since $R^*$ is a concave function,  the revenue-maximizing~$w$ can be determined by the Golden ratio or the ternary search as in~\citet{hardt2016equality}.
The technical difficulty is to bound the number of steps needed for the binary search to converge even though a small change in $\lambda_a$ can result in a massive change both in welfare and revenue.

{\fed Now we discuss all the steps in detail.
To simplify formulas,} we can assume w.l.o.g. that $\max_{x,a} \overline{u}(x,a)=\max_{x,a} |r(x,a)|= 1$ by changing the scale.

Proposition~\ref{lm_approx_marginal} below guarantees that for each subgroup $\{A=a\}$ and welfare level $w$, a marginal classifier with welfare $w$  and revenue at least $R_a^*(w)-\frac{\varepsilon}{2}$ can be computed in $O\left(\log\left(\frac{1}{\varepsilon}\right)\right)$ time by binary search over thresholds $\lambda_a$. We verify below that the upper bound on $\varepsilon$ ensures that none of $\{\lambda_0,\lambda_1\}$  exceeds $\frac{2}{\varepsilon}$  (a technical condition of Proposition~\ref{lm_approx_marginal}).

Therefore, for a given $w$ we can also compute the optimal total revenue $R^*(w)=\P(A=0)R_0^*(w)+\P(A=1)R_1^*(w)$ up to $\frac{\varepsilon}{2}$ in $O\left(\log\left(\frac{1}{\varepsilon}\right)\right)$. The total revenue $R^*$ is a continuous concave function of the welfare level $w$ (see the proof of Proposition~\ref{prop_optimal_GE}) and thus it can be approximately maximized by the Golden-ratio search or by the ternary search as in~\citet{hardt2016equality}.

For the moment, assume that we can compute $R^*(w)$ exactly. Then 
  $O\left(\log\frac{1}{\varepsilon}\right)$ iterations of the search will find  $w$ such that $R^*(w)\geq R^*(w^*)-\frac{\varepsilon}{2}$, where $w^*$ is the revenue-maximizing welfare level. Indeed,
the function $R^*(w)$ is defined on the interval $0\leq w\leq b$, where  $b=\min\{\E[u\vert A=0],\, \E[u\vert A=1]\}$, and takes values within $[-1,1]$ by the assumption on the magnitude of $r$. On any interval $[w,w']$ with $ |w'-w|\leq \frac{\varepsilon}{4}\cdot b$  containing the {\fed bank-optimal} welfare level $w^*$, the value $R^*(w^*)$ is within $\frac{\varepsilon}{2}$ of $\max\{R^*(w),R^*(w')\}$ (otherwise, by convexity, there exists $w''\in[0,b]$ such that $R^*(w'')<-1$); thus, $O\left(\log\frac{1}{\varepsilon}\right)$ steps of the search are enough.

If, along the way, the values of $R^*$ are computed with precision~$\frac{\varepsilon}{2}$, the output  $w$  of the Golden-ratio search satisfies $R^*(w)\geq R^*(w^*)-{\varepsilon}$ and, therefore, provides an {\fed $\varepsilon$-bank-optimal} WE-classifier. Since all $O\left(\log\frac{1}{\varepsilon}\right)$ iterations of the search algorithm involve computing $R^*$ at a new point, the total run-time is $O\left(\log\frac{1}{\varepsilon}\cdot \log\frac{1}{\varepsilon}\right)$.

In order to satisfy the technical condition of Proposition~\ref{lm_approx_marginal}, we need one more twist: the  search should start on the smaller interval
$$I=\Big[\max_{a\in\{0,1\}} w_\frac{2}{\varepsilon}(a), \, \min_{a\in\{0,1\}} w_{-\frac{2}{\varepsilon}}(a)\Big]\cap [0,b],$$
where $w_\lambda(a)$ denotes a point $w$ such that the super-gradient of $R^*_a(w)$ contains $\lambda$ (a point $w_\lambda(a)$ can be easily  computed as in the proof of Proposition~\ref{lm_approx_marginal}). By the definition of $I$, any $w\in I$ satisfies the condition of Proposition~\ref{lm_approx_marginal} and it remains to check that the Golden-ratio search, restricted to $I$, still approximates the optimum $w^*$. Put differently, it is enough to show that $w^*\in I$. At $w^*$, First-order conditions imply $0\in \partial R^*(w^*)=\P(A=0)\partial R^*_0(w^*)+\P(A=1)\partial R^*_1(w^*)$. Therefore, there exist $\lambda_a\in \partial R^*_a(w^*)$ such that $\P(A=0)|\lambda_0|=\P(A=1)|\lambda_1|$. Pick $\lambda_a$ that is non-negative. Then $\lambda_a\cdot w^*\leq 2$; otherwise, $R^*_a(0)\leq R^*_a(w^*)+\lambda_a (0-w^*)<-1$ which contradicts the normalization of $r$. By Lemma~\ref{lm_no_mismatch}, $w^*\geq W_{u,\UC}(0)$; hence, $\lambda_a\leq \frac{2}{W_{u,\UC}(0)}$. The upper bound on $\varepsilon$ implies that both $|\lambda_0|$ and $|\lambda_1|$ are bounded by $\frac{1}{3\varepsilon}$ and thus $w^*$ belongs to~$I$.
\end{proofof}        

The next proposition is an auxiliary technical result used in the proof of Proposition \ref{prop_omniscient_algorithm}. It shows that approximately {\fed bank-optimal} marginal classifier for a group $\{A=a\}$ with a given welfare level $w$ can be found by binary search over thresholds. The threshold $\lambda_a$ can also be thought as a super-gradient of the marginal revenue $R^*_a$. The super-gradient of a concave function may ``explode'' if the point is close to the boundary. This is why we need an additional technical assumption that the super-gradient of $R^*_a$ at $w$ is not too big. 
\begin{proposition}\label{lm_approx_marginal}
Under the assumptions of Proposition~\ref{prop_omniscient_algorithm}, for a given $a\in\{0,1\}$,  $\varepsilon>0$, and $w\in \big[0,\, \E[u\vert A=a]\big]$ such that the super-gradient $\partial R^*_a(w)$ contains $\lambda$ with $|\lambda|<\frac{1}{\varepsilon}$, a marginal classifier for a subgroup $\{A=a\}$  with welfare level $w$ and revenue of at least $R_a^*(w)-\varepsilon\cdot \max_{x,a} |r(x,a)|$ can be computed in $O\left(\log\left(\frac{1}{\varepsilon}\right)\right)$.
\end{proposition}
\begin{proofof}{Proposition \ref{lm_approx_marginal}} Similarly to the proof of Proposition~\ref{prop_omniscient_algorithm}, we can assume w.l.o.g. that $\max_{x,a} |r(x,a)|= 1$.  Consider the threshold classifier $c_\lambda$ such that $c_\lambda(x)=1$ if  $r(x,a)\geq \lambda \overline{u}(x,a)$ and equals zero otherwise. For any $\lambda\in \R$, this classifier generates the optimal subgroup revenue among all the classifiers having the same welfare level $w_\lambda(a)=\E[\overline{u}(X,A)c_\lambda(X)\mid A=a]$, i.e., $R_a(c_\lambda)$ equals $R_a^*(w_\lambda)$; moreover, $\lambda$ belongs to the super-gradient of the concave function $R^*_a$ at $w_\lambda$ (see the proof of Proposition~\ref{prop_optimal_GE}). 

Note that $w_\lambda$ is a decreasing function of $\lambda$ and thus we can use binary search in order to find a $\lambda$ such that $w_\lambda$ is close to $w$. By the assumption on $w$ we can restrict the search to $\lambda\in \left[-\frac{1}{\varepsilon},\frac{1}{\varepsilon}\right]$ and thus after $O\left(\log\frac{1}{\varepsilon}\right)$ steps find $\lambda$ and $\lambda'$ such that $w_\lambda\leq w\leq w_{\lambda'}$ and $|\lambda-\lambda'|\leq \varepsilon$. We can treat $w$ as the convex combination $\theta w_\lambda+(1-\theta)w_{\lambda'}$, $\theta\in[0,1]$, and consider a classifier $c=\theta c_\lambda+(1-\theta)c_{\lambda'}$. By the construction, $c$ has the right welfare level of $w$. 
 
Let us check that $R_a(c)\geq R^*_a(w)-\varepsilon$.  For any concave function $h=h(z)$, $z_1,z_2\in \R$, and $\theta\in[0,1]$, it holds that
\[
\theta h(z_1)+(1-\theta)h(z_2)\geq h(\theta z_1+(1-\theta)z_2)-|z_1-z_2|\cdot|g_1-g_2|, \]
where $g_k$ is any super-gradient of $h$ at $z_k$. The revenue of $c$ is given by the convex combination of values $\theta R^*_a(w_\lambda)+(1-\theta)R^*(w_{\lambda'})$, while $R^*_a(w)$ is the value at the convex combination $w=\theta w_\lambda+(1-\theta)w_{\lambda'}$.
 Applying the above inequality, we get 
 $R_a(c)\geq R^*_a(w)-|w_\lambda-w_{\lambda'}|\cdot|\lambda-\lambda'|$; hence, $R_a(c)\geq R^*_a(w)-\varepsilon$.
\end{proofof}

\subsection{Incomplete information}\label{subsect_appp_com}

The assumption of expectation oracle from Proposition~\ref{prop_omniscient_algorithm}, as well as the exact knowledge of revenues/losses $\alpha_+$ and $\alpha_-$ can be relaxed. {\fed Moreover, in this subsection we allow for a mismatch between the unknown underlying utility function $v$ and the known utility $u$ assumed by the regulator. We assume that the mismatch is small and interpret $u$ as the estimator of $v$.}
 
The following proposition mirrors Proposition \ref{prop:sample_and_estimators}.
\begin{proposition}[Full version of Proposition~\ref{prop:sample_and_estimators}]\label{prop:sample_and_estimators_appendix} 
Fix $\delta>0$ and assume that the bank has access to a sample of $(X,A,Y,\alpha_\pm,v)$ and to estimators {\fed $u={{u}}(x,a)$ of $v$ and $\hat{r}=\hat{r}(x,a)$ of $r$ such that $\E\left[ \abs{{{\overline u}}-\overline v}\right]\leq \eta_u \cdot \max_{x,a} \overline{v}(x,a)$} and $\E\left[\abs{\hat{r}-r}\right]\leq \eta_r\cdot \max_{x,a} |{r}(x,a)|$. Let $\varepsilon$ be such that $\varepsilon>2\sqrt{6\left(\frac{1}{\P(A=0)}+\frac{1}{\P(A=1)}\max\{\eta_u,\eta_r\}\right)}$ that also satisfies the upper bound in Inequality~\eqref{eq_varepsilon}. Then, an $\Big(\varepsilon\cdot \max_{x,a} |{r}(x,a)|,\ \ \varepsilon\cdot \max_{x,a} \overline{v}(x,a)\Big)$ {\fed bank-optimal $v$-WE} classifier
can be computed with probability of at least $1-\delta$ on a sample of size 
$O\left(\frac{1}{\varepsilon^2}\left(\log\frac{1}{\delta}+\log\log\frac{1}{\varepsilon}\right)\right)$.
\end{proposition}

The proof of the proposition is deferred to the end of the subsection, and it relies on three auxiliary propositions. First, in Proposition~\ref{lm_sample_complexity} we prove that having a sample access to the distribution is enough for constructing an $(\varepsilon,\varepsilon)$ {\fed bank-optimal} $v$-WE classifier of a particular functional form.

\begin{definition}
A classifier $c$ is a \emph{classifier in the reduced form}\footnote{Note that both {\fed bank-optimal} and {\fed $\varepsilon$-bank-optimal} classifiers from Propositions~\ref{prop_optimal_GE} and~\ref{prop_omniscient_algorithm} have a reduced form.} if $c=c(r(x,a),\,\overline{v}(x,a))$  (i.e., $c$ depends on $x$ and $a$ only through $r(x,a)$ and $\overline{v}(x,a)$).
\end{definition}
{\fed Note that in order to compute the outcome of a reduced form classifier, one must know the exact values of $r$ and $v$; we use their approximations instead.
However, the classifier from Proposition~\ref{lm_sample_complexity} turns out to be sensitive to small errors in $r$ and ${v}$. In particular, using the estimator $u$ instead of the exact underlying utility $v$ can dramatically change the outcome.}  In Proposition~\ref{lm:robust}, we describe the smoothing technique that allows us to ensure that the outcome of classification is robust to small perturbations of $r$ and $v$ in the following sense.
\begin{definition} A reduced-form classifier $c$ is $\varkappa$-\emph{robust} if
$$|c(s,w)-c(s',w')|\leq \frac{1}{\varkappa}\left(\frac{|s-s'|}{\max_{x,a} |r(x,a)|}+\frac{|{w}-{w}'|}{\max_{x,a} \overline{v}(x,a)}\right).$$
{\fed for all numbers $s,s'\in \R$ and $w,w'\in \R_+$.}
\end{definition}
Finally, Proposition~\ref{lm:robustness_and_estimators} bounds the losses in revenue and welfare of a $\varkappa$-robust classifier that relies on estimators of $v$ and $r$ instead of the exact values.

\begin{proposition}\label{lm_sample_complexity}
For $\delta>0$ and $\varepsilon>0$ satisfying Inequality~\eqref{eq_varepsilon}, a reduced form of an $\left(\varepsilon\cdot \max_{x,a} |{r}(x,a)|, \varepsilon\cdot \max_{x,a} \overline{v}(x,a)\right)$ {\fed bank-optimal} $v$-WE classifier
can be computed with probability of at least $1-\delta$ on a sample of size $O\left(\frac{1}{\varepsilon^2}\left(\log\frac{1}{\delta}+\log\log\frac{1}{\varepsilon}\right)\right)$.
\end{proposition}
\begin{proofof}{Proposition \ref{lm_sample_complexity}}
We mimic the algorithm from Proposition~\ref{prop_omniscient_algorithm} but change all the exact expectations to empirical averages.

By re-scaling, we assume $\max_{x,a}|r(x,a)|=\max_{x,a} \overline{v}(x,a)=1$. The algorithm from Proposition~\ref{prop_omniscient_algorithm} computes the following conditional expectations $O\left(\log^2\frac{1}{\varepsilon}\right)$ times: welfares  $\E[{v}(X,A,Y)c_\lambda(X)\mid A=a]$ and revenues $\E[r(X,A)c_\lambda(X)\mid A=a]=\E[(\alpha_+(X)Y-\alpha_-(X)(1-Y))c_\lambda(X)\mid A=a]$. By the Chernoff bound, for i.i.d. random variables $\xi_1,...,\xi_K$ in $[-1,1]$ we have $\P\left(\left|\E\xi_1-\frac{1}{K}\sum_{k=1}^K \xi_k\right|> \varepsilon \right)\leq 2\exp\left(-\frac{K\varepsilon^2}{2}\right)$. Hence, we ensure that each particular expectation is computed with accuracy $\varepsilon$ with probability $1-\delta'$ using a sample of size $O\left(\frac{1}{\varepsilon^2}\log \frac{1}{\delta'}\right)$ conditional on $\{A=a\}$. In order to have such a conditional sample,  the unconditional sample must be $\frac{1}{\P(A=a)}$-times bigger.
By the union bound, selecting $\delta'$ such that $\delta'\cdot O\left(\log^2\frac{1}{\varepsilon}\right)\leq \delta$ we guarantee that with probability of at least $1-\delta$ all $O\left(\log^2\frac{1}{\varepsilon}\right)$  expectations are computed with precision $\varepsilon$ throughout the run of the algorithm.
\end{proofof}

\begin{proposition}\label{lm:robust}
The result of Proposition~\ref{lm_sample_complexity} remain true if we additionally ask the computed classifier to be $\frac{\varepsilon}{3}$-robust.
\end{proposition}

\begin{proofof}{Proposition \ref{lm:robust}}
We describe the smoothing technique as usual assuming $\max_{x,a}|r(x,a)|=\max_{x,a}|\overline{v}(x,a)|=1$. To compute an $\frac{\varepsilon}{3}$-robust classifier $c$, we consider an auxiliary ``smoothed'' classification problem with $x^\circ=(x,\xi_v,\xi_r)$, where $\xi_v,\xi_r$ are uniformly distributed on $[0,\varepsilon/3]$ and are independent from each other and from $(x,a)$, and define $\overline{v}^\circ(x^\circ, a)=\overline{v}(x,a)+\xi_v$ and $r^\circ(x',a)=r(x,a)+\xi_r$. 

Next, we compute the reduced form of an $\left(\varepsilon/3, \varepsilon/3\right)$ {\fed bank-optimal} ${v}^\circ$-WE classifier $c^\circ$ in the new problem  and define a reduced-form classifier in the original problem (under $v$) by $c(r,\overline{v})=\E_{(\xi_v,\xi_r)} c^\circ\big(r+\xi_r, \overline{v}+\xi_v\big)$. If we slightly perturb $r$ and $\overline{v}$,  then $c^\circ$ is integrated over almost-coinciding squares in the $(r,\overline{v})$-plane and hence $|c(r,\overline{v})-c(r',\overline{v}')|\leq \frac{|r-r'|+|\overline{v}-\overline{v}'|}{\varepsilon/3}$, i.e., we get robustness. The constructed classifier $c$ equalizes welfare up to $2\varepsilon/3\leq \varepsilon$  since $v$ and $v^\circ$ differ by at most $\varepsilon/3$. 

It remains to check that $c$ generates $\varepsilon$-optimal revenue. The {\fed bank-optimal} classifier $c^*_{\mathrm{WE}(v)}$ in the original problem induces a classifier that equalizes welfare up to $\varepsilon/3$ and has revenue $R^\circ(c^*_{\mathrm{WE}(v)})\geq R(c^*_{\mathrm{WE}(v)})-\varepsilon/3$ in the new problem. We get  $R^\circ(c^\circ)+\varepsilon/3\geq R^\circ(c^*_{\mathrm{WE}(v)})\geq R(c^*_{\mathrm{WE}(v)})-\varepsilon/3$. Since $R(c)\geq R^\circ(c^\circ)-\varepsilon/3$, we obtain that $c$ generates $\varepsilon$-optimal revenue in the original problem.
\end{proofof}    

\begin{proposition}\label{lm:robustness_and_estimators} Suppose we have estimators ${{u}}$ and $\hat{{r}}$ such that $\E\left[|{\overline{u}}-\overline v|\right]\leq \eta_u\cdot \max_{x,a} v(x,a)$ and $\E\left[|\hat{r}-r|\right]\leq \eta_r\cdot \max_{x,a} |r(x,a)|$. Further, let $c$ be a reduced form of a $\varkappa$-robust $(\varepsilon,\varepsilon')$ {\fed bank-optimal} WE classifier and define $c'$ as the version of $c$ that uses the estimators ${{u}}$ and $\hat r$, i.e., $c'(x,a)=c(\hat r, {\overline{u}})$. Then, $c'$ is 
$\left(\hat{\varepsilon},\hat{\varepsilon'}\right)$ {\fed bank-optimal} with $$\hat{\varepsilon}=\varepsilon+\frac{\max_{x,a} |r(x,a)|}{\varkappa}\left(\eta_u+\eta_r\right)$$ and
$$\hat{\varepsilon'}=\varepsilon'+ \frac{\max_{x,a} {{v}(x,a)}}{\varkappa}\left(\eta_u+\eta_r\right)\left(\frac{1}{\P(A=0)}+\frac{1}{\P(A=1)}\right).$$
\end{proposition}

\begin{proofof}{Proposition \ref{lm:robustness_and_estimators}}
By $\varkappa$-robustness we get
\begin{align*}
\left|\E \left[v(x,a) c(\hat{r},{\overline{u}})\mid A=a\right]-\E \left[v(x,a) c({r},{\overline{v}})\mid A=a\right]\right|
&\leq\max_{x,a} v(x,a)\E\left[\left|c(\hat{r},{\overline{u}})- c({r},{\overline{v}})\right|\mid A=a\right] \\
&\leq\max_{x,a} v(x,a)\frac{1}{\P(A=a)}\E\left|c(\hat{r},{\overline{u}})- c({r},{\overline{v}})\right|  \\
& \leq\frac{\max_{x,a} v(x,a)}{\varkappa\cdot \P(A=a)}\left(\eta_u+\eta_r\right); 
\end{align*}
thus, by applying the triangle inequality we get 
\[
\left|W_{u,c(\hat{r},{\overline{u}})}(1)-W_{u,c(\hat{r},{\overline{u}})}(0)\right|\leq \varepsilon'+\frac{\max_{x,a} {{v}(x,a)}}{\varkappa}\left(\eta_u+\eta_r\right)\left(\frac{1}{\P(A=0)}+\frac{1}{\P(A=1)}\right).
\]
The argument for the revenue is similar but simpler since we need unconditional expectations only.
\end{proofof}

We are now ready to prove Proposition~\ref{prop:sample_and_estimators_appendix}. 
\begin{proofof}{Proposition~\ref{prop:sample_and_estimators_appendix}}
As usual assume $\max_{x,a} \overline{v}(x,a)=\max_{x,a} {r}(x,a)=1$.
By Proposition~\ref{lm:robust}, we can compute an $\frac{\varepsilon}{6}$-robust $\left(\frac{\varepsilon}{2},\frac{\varepsilon}{2}\right)$ {\fed bank-optimal} WE classifier $c$ in the reduced form on a sample of required size. The lower bound on $\varepsilon$ is chosen in such a way that 
$\frac{6}{\varepsilon}(\eta_u+\eta_r)\left(\frac{1}{\P(A=0)}+\frac{1}{\P(A=1)}\right)\leq \frac{\varepsilon}{2}$ and thus, by Proposition~\ref{lm:robustness_and_estimators}, the classifier $c(\hat{r}(x,a),{\overline{u}}(x,a))$ is an $(\varepsilon,\varepsilon)$ {\fed bank-optimal} WE classifier.
\end{proofof}

}\fi}

\end{document}